\documentclass[10pt, twocolumn, journal]{IEEEtran}%
\ifCLASSINFOpdf
\else
\fi
\usepackage{amsmath,graphicx}
\newcommand{\mb}[1]{\mathbf{#1}}
\newcommand{\tr}{\mathop{\mathrm{tr}}}

\usepackage{braket}
\usepackage{mathtools}
\usepackage{amsmath}
\usepackage{amssymb}
\usepackage{amsfonts}
\usepackage{lipsum}
\usepackage{mathtools}
\usepackage{lscape}
\usepackage{multirow}
\usepackage{here}
\usepackage{bm}
\usepackage{cuted}
\usepackage{adjustbox}
\usepackage{cite}

\usepackage{makeidx}

\usepackage{setspace}

\usepackage{color}

\usepackage{here}
\usepackage{subfigure}
\usepackage{amsthm}
\newtheorem{definition}{Definition}
\newtheorem{theorem}{Theorem}

\newtheorem{corollary}{Corollary}
\newtheorem{lemma}{Lemma}

\theoremstyle{remark}

\usepackage{hhline}

\usepackage{xcolor,colortbl}

\definecolor{Gray}{gray}{0.9}
\definecolor{LightCyan}{rgb}{0.88,1,1}
\definecolor{LightMagenta}{rgb}{1, 0.88, 1}
\definecolor{LightBlue}{rgb}{0.67, 0.85, 0.9}

\newcolumntype{A}{>{\columncolor{LightCyan}}c}
\newcolumntype{B}{>{\columncolor{LightMagenta}}c}
\newcolumntype{C}{>{\columncolor{LightBlue}}c}
\newcolumntype{D}{>{\columncolor{LightCyan}}r}
\newcolumntype{E}{>{\columncolor{LightMagenta}}l}
\newcolumntype{F}{>{\columncolor{Gray}}c}

\usepackage{array}     
\usepackage{longtable}
\usepackage{colortab}
\usepackage{colortbl}

\usepackage{booktabs}%

\makeatletter
\newcommand{\pushright}[1]{\ifmeasuring@#1\else\omit\hfill$\displaystyle#1$\fi\ignorespaces}
\newcommand{\pushleft}[1]{\ifmeasuring@#1\else\omit$\displaystyle#1$\hfill\fi\ignorespaces}

\makeatletter
\newcommand*\bigcdot{\mathpalette\bigcdot@{.5}}
\newcommand*\bigcdot@[2]{\mathbin{\vcenter{\hbox{\scalebox{#2}{$\m@th#1\bullet$}}}}}
\makeatother

\makeatletter
\newcommand*{\itemequation}[3][]{%
  \item
  \begingroup
    \refstepcounter{equation}%
    \ifx\\#1\\%
    \else  
      \label{#1}%
    \fi
    \sbox0{#2}%
    \sbox2{$\displaystyle#3\m@th$}%
    \sbox4{\@eqnnum}%
    \dimen@=.5\dimexpr\linewidth-\wd2\relax
    \ifcase
        \ifdim\wd0>\dimen@
          \z@
        \else
          \ifdim\wd4>\dimen@
            \z@
          \else 
            \@ne
          \fi 
        \fi
      \@latex@warning{Equation is too large}%
    \fi
    \noindent   
    \rlap{\copy0}%
    \rlap{\hbox to \linewidth{\copy2\hfill}}%
    \hbox to \linewidth{\hfill\copy4}%
    \hspace{0pt}%
  \endgroup
  \ignorespaces 
}
\makeatother

\makeatletter
\newcommand*{\nonumitemequation}[3][]{%
  \item[]
  \begingroup
    \refstepcounter{equation}%
    \ifx\\#1\\%
    \else  
      \label{#1}%
    \fi
    \sbox0{#2}%
    \sbox2{$\displaystyle#3\m@th$}%
    \sbox4{\@eqnnum}%
    \dimen@=.5\dimexpr\linewidth-\wd2\relax
    \ifcase
        \ifdim\wd0>\dimen@
          \z@
        \else
          \ifdim\wd4>\dimen@
            \z@
          \else 
            \@ne
          \fi 
        \fi
      \@latex@warning{Equation is too large}%
    \fi
    \noindent   
    \rlap{\copy0}%
    \rlap{\hbox to \linewidth{\copy2\hfill}}%
    \hbox to \linewidth{\hfill\copy4}%
    \hspace{0pt}%
  \endgroup
  \ignorespaces 
}
\makeatother

\title{Graph Signal Sampling Under Stochastic Priors}

\author{Junya~Hara,~\IEEEmembership{Student Member,~IEEE,}
        Yuichi~Tanaka,~\IEEEmembership{Senior Member,~IEEE,}
        and Yonina~C.~Eldar,~\IEEEmembership{Fellow,~IEEE}%
\thanks{Preliminary results of this work was presented in \cite{hara_generalized_2020}.}
\thanks{JH and YT are supported in part by JST PRESTO (JPMJPR1656) and JSPS KAKENHI (19K22864). YE has received funding from European Research Council (ERC) under the European Union's Horizon 2020 research and innovation program (grants no. $278025{,}677909{,}681514$).}}
\begin{document}

\maketitle
\IEEEpeerreviewmaketitle

\begin{abstract}

We propose a generalized sampling framework for stochastic graph signals. Stochastic graph signals are characterized by
graph wide sense stationarity (GWSS) which is an extension of wide sense stationarity (WSS) for standard time-domain signals.
In this paper, graph signals are assumed to satisfy the GWSS conditions and we study their sampling as well as recovery procedures.
In generalized sampling, a correction filter is inserted between sampling and reconstruction operators to compensate for non-ideal measurements. 
We propose a design method for the correction filters to reduce the mean-squared error (MSE) between original and reconstructed graph signals.
We derive the correction filters for two cases: The reconstruction filter is arbitrarily chosen or predefined.
The proposed framework allows for arbitrary sampling methods, i.e., sampling in the vertex or graph frequency domain.
We also show that the graph spectral response of the resulting correction filter parallels that for generalized sampling for WSS signals if sampling is performed in the graph frequency domain.
Furthermore, we reveal the theoretical connection between the proposed  and existing correction filters.
The effectiveness of our approach is validated via experiments by comparing its MSE with existing approaches.

\end{abstract}


\section{Introduction}\label{sec:intro}
Graph signal processing (GSP) is a developing field in signal processing \cite{shuman_emerging_2013,ortega_graph_2018}.
Applications of GSP are extensive, including learning of graphs\cite{dong_learning_2016,kalofolias_how_nodate,egilmez_graph_2018,yamada_time-varying_2020,dong_learning_2019}, restoration of graph signals\cite{ono_total_2015,onuki_graph_2016,nagahama_graph_2021,chen_sampling_2020}, image/point cloud processing\cite{cheung_graph_2018,wang_dynamic_2019}, and graph neural networks\cite{bronstein_geometric_2017}.
Recent interest in GSP is to extend classical signal processing theory to the graph setting \cite{agaskar_spectral_2013,tanaka_generalized_2020,sandryhaila_big_2014,narang_perfect_2012,shuman_multiscale_2016}.
One of the main differences between standard signal processing and GSP is that GSP systems are not shift-invariant (SI) in general.
This leads to the challenge that GSP systems may have different definitions in the vertex and spectral (graph frequency) domains: They do not coincide in general.
Such examples include sampling \cite{tanaka_sampling_2020,anis_efficient_2016,chen_discrete_2015,marques_sampling_2016,puy_random_2018,tsitsvero_signals_2016,sakiyama_eigendecomposition-free_2019,chamon_greedy_2018,tanaka_spectral_2018}, translation \cite{shuman_vertex-frequency_2016,sandryhaila_discrete_2014,girault_stationary_2015,dees_unitary_2019,gavili_shift_2017}, and filtering \cite{shi_graph_2019}.

In this paper, we focus on graph signal sampling.
Sampling theory for graph signals has been widely studied in GSP \cite{pesenson_sampling_2008,tanaka_sampling_2020,anis_efficient_2016,chen_discrete_2015,marques_sampling_2016,puy_random_2018,wang_-optimal_2018,tsitsvero_signals_2016,sakiyama_eigendecomposition-free_2019,chamon_greedy_2018}. 
There are many promising applications of graph signal sampling, such as sensor placement \cite{sakiyama_eigendecomposition-free_2019}, filter bank designs\cite{hammond_wavelets_2011,narang_compact_2013,tanaka_m_2014,sakiyama_two-channel_2019}, traffic monitoring \cite{chen_monitoring_2016}, and semi-supervised learning\cite{gadde_active_2014,anis_sampling_2019}.
Most works focus on building parallels of the Shannon-Nyquist theorem \cite{jerri_shannon_1977,unser_sampling50_2000} and its generalizations \cite{eldar_sampling_2015,eldar_beyond_2009} to the graph setting.
Therefore, sampling of bandlimited graph signals has been widely studied  \cite{anis_efficient_2016,chen_discrete_2015,tsitsvero_signals_2016,chamon_greedy_2018}.
Other graph signal subspaces are also useful for practical applications: For example, piecewise smooth graph signals and periodic graph spectrum signals have been considered \cite{li_scalable_2019,chen_multiresolution_2018,varma_passive_2019,chen_representations_2016}.

Since sampling of deterministic signals has been well studied in the literature of sampling in Hilbert space\cite{eldar_genral_2005,christensen_oblique_2004,eldar_beyond_2009,eldar_sampling_2015}, we can immediately derive its GSP counterpart.
In contrast, sampling of \textit{random graph signals} has not been considered so far because existing (generalized) sampling methods for graph signals have mainly focused on deterministic signal models \cite{tanaka_sampling_2020,anis_efficient_2016,chen_discrete_2015,marques_sampling_2016,puy_random_2018,tsitsvero_signals_2016,sakiyama_eigendecomposition-free_2019,chamon_greedy_2018,tanaka_generalized_2020,chepuri_graph_2018}.

In this paper, we consider a graph signal sampling framework for random graph signals.
Random graph signals are modeled by graph wide sense stationarity (GWSS), which is a counterpart of wide sense stationarity (WSS) of standard signals.
WSS is characterized by the statistical moments, i.e., mean and covariance, that are invariant to shift.
As previously mentioned, graph signals do not lie in SI spaces in general.
Hence, existing definitions of GWSS are based on the spectral characteristics of signals\cite{girault_stationary_2015,perraudin_stationary_2017,marques_stationary_2017}, which are extensions of the power spectral density (PSD) of the standard WSS.
However, existing definitions of GWSS require additional assumptions that are not necessary for WSS.
Therefore, we first define GWSS as a natural extension of WSS based on signal modulation.

Subsequently, we develop a generalized sampling framework to best recover GWSS signals.
Our framework 
parallels generalized sampling of SI signals \cite{eldar_nonideal_2006}.
It consists of sampling, correction, and reconstruction transforms.
The correction transform is inserted between the sampling and reconstruction transforms to compensate for non-ideal measurements.
We derive the correction transform that minimizes the mean-squared error (MSE). 
Our framework can be applied to any sampling method that is linear.
In other words, both vertex and graph frequency domain sampling methods\cite{anis_efficient_2016,chen_discrete_2015,tsitsvero_signals_2016,chamon_greedy_2018,sakiyama_eigendecomposition-free_2019,tanaka_spectral_2018} are applicable without changing the framework.

Existing works of (generalized) sampling theory on graphs  \cite{chepuri_graph_2018,tanaka_generalized_2020} have been studied only for deterministic signal models.
In contrast, our generalized sampling is designed for random graph signals as described in Table \ref{tab:field_gen_samp}.
Interestingly, our solution parallels that in the SI setting \cite{eldar_nonideal_2006} when sampling is performed in the graph frequency domain.
Moreover, we reveal that the existing signal recovery under different priors \cite{tanaka_generalized_2020} is special cases of our proposed recovery.
Experiments for synthetic signals validate that our proposed recovery is effective for stochastic graph signals with sampling in both vertex and spectral domains.

The remainder of this paper is organized as follows. Section \ref{sec:sampling} reviews generalized sampling for time-domain signals. For WSS signals, we introduce the standard  Wiener filter. In Section \ref{sec:samp_stationarity}, generalized sampling framework for stochastic graph signals is introduced. In addition, GWSS is defined as an extension of WSS. Section \ref{sec:graph_wiener} derives graph Wiener filters for recovery of a GWSS process based on the minimization of the MSE between the original and reconstructed graph signals. In the special case of graph frequency domain sampling, the graph frequency response parallels that of sampling in SI spaces. In Section \ref{sec:relation_priors}, we discuss the relationship among the proposed Wiener filter and existing generalized graph signal sampling. Signal recovery experiments for synthetic and real-world signals are demonstrated in Section \ref{sec:experiments}. Section \ref{sec:conclusion} concludes the paper.

\textit{Notation:} We consider a weighted undirected graph $\mathcal{G}=(\mathcal{V,E})$, where $\cal V$ and $\cal E$ represent the sets of vertices and edges, respectively. The number of vertices is $N = |\mathcal{V}|$ unless otherwise specified. The adjacency matrix of $\mathcal{G}$ is denoted by $\mathbf{A}$ where its $(m,n)$-element $a_{mn}\geq 0$ is the edge weight between the $m$th and $n$th vertices; $a_{mn}=0$ for unconnected vertices. The degree matrix $\mathbf{D}$ is defined as $\mathbf{D}=\text{diag}\,(d_{0},d_{1},\ldots,d_{N-1})$, where $d_{m}=\sum_na_{mn}$ is the $m$th diagonal element. We use a graph Laplacian $ \mathbf{L}\coloneqq \mathbf{D}-\mathbf{A}$ as a graph variation operator. A graph signal $\bm{x} \in \mathbb{C}^N$ is defined as a mapping from the vertex set to the set of complex numbers, i.e., $\bm{x}:\mathcal{V}\rightarrow \mathbb{C}$.

\begin{table}[t!]
\centering
\caption{Comparison of generalized graph signal sampling. }\label{tab:field_gen_samp}
\begin{tabular}{c||c|c}
\hline
Signal priors  & \multicolumn{1}{c}{SI spaces} & \multicolumn{1}{|c}{Graph subspaces} \\ \hhline{===}
Subspace   & \multicolumn{1}{c}{\cite{eldar_sampling_2015}} & \multicolumn{1}{|c}{\cite{chepuri_graph_2018,tanaka_generalized_2020}} \\
Smootheness  & \multicolumn{1}{c}{\cite{eldar_sampling_2015}} & \multicolumn{1}{|c}{\cite{chepuri_graph_2018,tanaka_generalized_2020}} \\
Stochatic  & \multicolumn{1}{c}{\cite{eldar_nonideal_2006}} & \multicolumn{1}{|c}{\textbf{This work}}  \\ \hline
\end{tabular}
\end{table}

The graph Fourier transform (GFT) of $\bm{x}$ is defined as
\begin{align}
\hat{x}(\lambda_i)=\braket{\bm{u}_{i},\bm{x}}=\sum_{n=0}^{N-1}u_{i}[n]x[n],
\end{align}
where $\bm{u}_{i}$ is the $i$th column of a unitary matrix $\mathbf{U}$ and it is obtained by the eigenvalue decomposition of the graph Laplacian $\mathbf{L}=\mathbf{U\Lambda U}^*$ with the eigenvalue matrix $\bm{\Lambda}=\mathrm{diag}\,(\lambda_0,\lambda_1,$ $\ldots,\lambda_{N-1})$.
We refer to $\lambda_i$ as a \textit{graph frequency}. 
We let $[\cdot]_{n,k}$ and $(\cdot)_{n}$ denote the ($n,k$)-element in the matrix and $n$-th element in the vector, respectively, $(*)$, $\braket{\cdot,\cdot}$ and $\tr(\cdot)$ denote the convolution, the inner product between two vectors and trace of a matrix, respectively, and $(\circ)$ represents the Hadamard (elementwise) product.

\begin{figure*}[t!]
\centering
  \includegraphics[width=1.8\columnwidth]{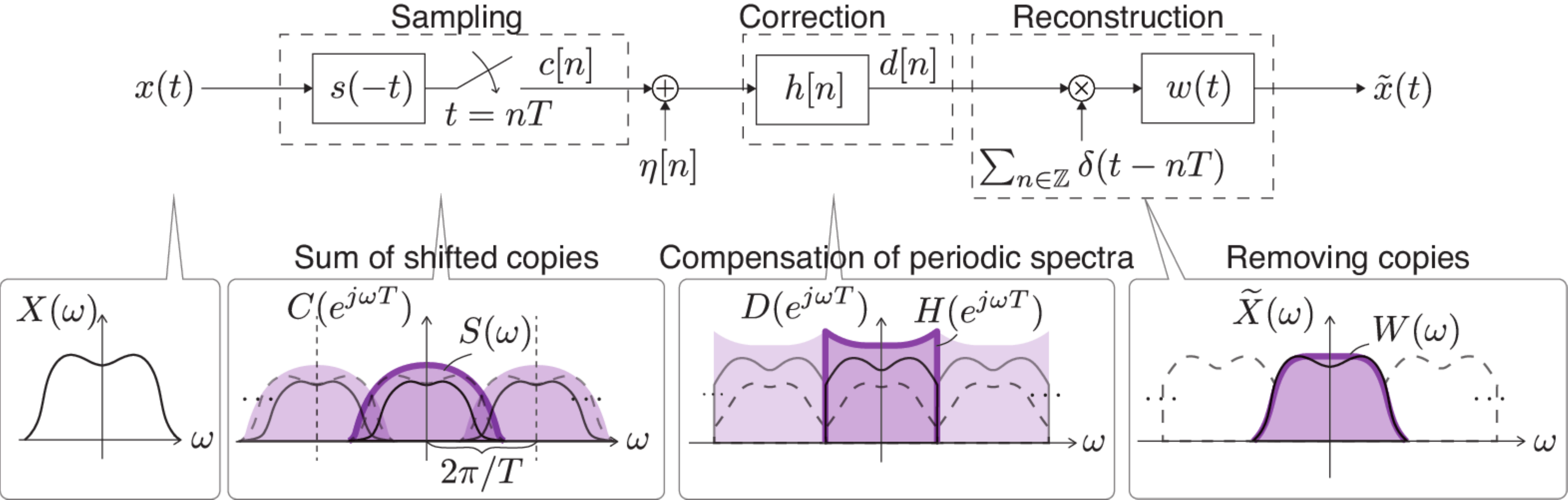}\vspace{-1em}
  \caption{Generalized sampling framework in SI spaces and its counterpart in the Fourier domain. Dashed-lined and solid-lined areas in the bottom represent input and filtered spectra, respectively. Colored areas represent spectral responses of each filter.}\label{framework_si}
\end{figure*}

For time-domain signals, $x(t)$ and $x[n]$ denote a continuous-time signal and discrete-time signal, respectively. The continuous-time Fourier transform (CTFT) of a signal $x(t)\in L_2$ is denoted by $X(\omega)$ and the discrete-time Fourier transform (DTFT) of a sequence $x[n]\in\ell_2$ is denoted by $X(e^{j\omega})$.

\section{Generalized sampling in SI spaces}\label{sec:sampling}

Generalized sampling for stochastic standard signals is reviewed in this section \cite{eldar_nonideal_2006}. 
Detailed derivations and discussions can be found in \cite{eldar_nonideal_2006,eldar_sampling_2015}.
The Wiener filter introduced in this section parallels that in our graph signal sampling framework.

\subsection{Sampling and recovery framework for time-domain signals}\label{sec:gen_sampling}

We first review standard generalized sampling for time-domain signals \cite{eldar_sampling_2015}. The framework is depicted in Fig. \ref{framework_si}. 

Let $x(t)$ be a continuous-time signal and $\eta[n]$ be stationary noise. We consider samples $c[n]$ at $t=nT$ of a filtered signal of $x(t)$, $c(t)=x(t)*s(-t)$
, where $s(-t)$ denotes a sampling filter. The DTFT of the samples, $C(e^{j\omega})$, can be written as
\begin{align}
C(e^{j\omega})=\frac{1}{T}\sum_{k\in\mathbb{Z}}S^*\left(\frac{\omega-2\pi k}{T}\right)X\left(\frac{\omega-2\pi k}{T}\right),\label{eq:std_freq_samp}
\end{align}
where $S(\omega)$ is the CTFT of $s(t)\in L_2$. 
The measured samples are corrupted by noise and are given by
\begin{align}
y[n]=\braket{s(t-nT),x(t)}+\eta[n].
\end{align}
The reconstructed signal is constrained to lie in a SI space $\mathcal{W}$, spanned by the shifts of a reconstruction kernel $w(t)\in L_2$, i.e., 
\begin{align}
\tilde{x}(t)=\sum_{n\in\mathbb{Z}}d[n]w(t-nT),\label{eq:std_time_samp}
\end{align}
where $d[n]\in \ell_2$ are unknown expansion coefficients, which yield the best possible recovery $\tilde{x}(t)$. 
To obtain an appropriate $d[n]$, we apply a digital correction filter $h[n]$ to $y[n]$, i.e., $d[n]=(h*y)[n]$. 
The CTFT of $\tilde{x}(t)$, $\widetilde{X}(\omega)$, can be expressed as
\begin{align}
\widetilde{X}(\omega)=D(e^{j\omega T})W(\omega),\label{eq:std_freq_reconst}
\end{align}
where $D(e^{j\omega})$ is the DTFT of $d[n]$ and $W(\omega)$ is the CTFT of $w(t)$.

In practice, sampling and reconstruction filters, i.e., $s(t)$ and $w(t)$, may be predefined prior to sampling based on technical requirements. For generalized sampling, instead, we assume that we can design the correction filter $h[n]$ freely.
The best correction filter $h[n]$ is designed such that the input signal $x(t)$ and the reconstructed signal $\tilde{x}(t)$ are close enough in some metric.
This is a widely-accepted sampling framework for time-domain signals and we follow it here for the graph setting \cite{eldar_minimum_2006,eldar_nonideal_2006}.

\begin{figure*}[t!]
\centering
  \includegraphics[width=1.5\columnwidth]{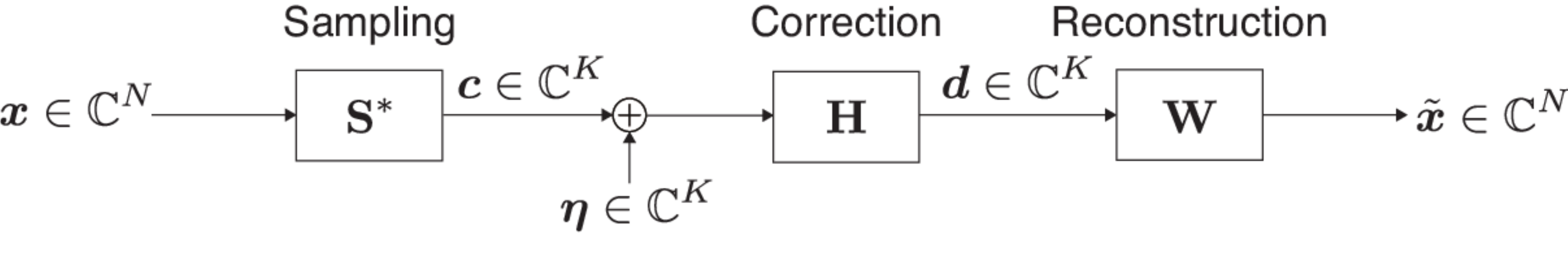}
  \caption{Generalized sampling framework for graph signals.}\label{gene_corrupt_filter_res}
\end{figure*}

\subsection{Wiener filter}

Next, suppose that $x(t)$ is a zero-mean WSS process with known power spectral density (PSD) $\Gamma_x(\omega)$ and $\eta[n]$ is a zero-mean WSS noise process with known PSD $\Gamma_\eta(e^{j\omega})$, independent of $x(t)$.
Detailed definitions of WSS are given in Appendix \ref{sec:stationarity}.

The optimal correction filter proposed in \cite{eldar_nonideal_2006} is obtained by solving the following problem:
\begin{align}
\min_{h[n]}\mathbb{E}\left\{\left|\tilde{x}(t)-P_{\mathcal{W}}x(t)\right|^2\right\},\label{eq:worst-case}
\end{align}
where $P_\mathcal{W}$ is an orthogonal projection onto the reconstruction subspace $\mathcal{W}$.
Since the recovered signal $\tilde{x}(t)$ is constrained by \eqref{eq:std_time_samp} to lie in $\mathcal{W}$, we want it to best approximate the orthogonal projection of $x(t)$ onto that same space $\mathcal{W}$ \cite{eldar_nonideal_2006}.
The frequency response of the solution of \eqref{eq:worst-case} is given by \cite{eldar_nonideal_2006}.

\begin{align}
&H_{\text{PRE}}(e^{j\omega})\nonumber \\
&=\frac{\sum_{k\in\mathbb{Z}}\Gamma_x(\frac{\omega}{T}+\frac{2\pi k}{T})S(\frac{\omega}{T}+\frac{2\pi k}{T})W^*(\frac{\omega}{T}+\frac{2\pi k}{T})}{R_{W}(e^{j\omega})\left(\Gamma_\eta(e^{j\omega})+\sum_{k\in\mathbb{Z}}\Gamma_x(\frac{\omega}{T}+\frac{2\pi k}{T})\left|S(\frac{\omega}{T}+\frac{2\pi k}{T})\right|^2\right)},\label{eq:stdwienercons}
\end{align}
where $R_{W}(e^{j\omega})=\sum_{k\in\mathbb{Z}}|W(\frac{\omega}{T}+\frac{2\pi k}{T})|^2$.

When the reconstruction filter is unconstrained, we can optimize the correction filter jointly with the reconstruction filter.
The optimal filter in the unconstrained setting is a special case of \eqref{eq:stdwienercons}, where $w(t)$ is optimally chosen. 
In order to obtain the unconstrained solution, we thus solve \eqref{eq:worst-case} with respect to $w(t)$ by substituting \eqref{eq:stdwienercons} for \eqref{eq:worst-case} \cite{eldar_sampling_2015}.
The resulting optimal correction filter in the unconstrained case is given by
\begin{align}
H_{\text{UNC}}(e^{j\omega})=\frac{1}{\Gamma_\eta(e^{j\omega})+\sum_{k\in\mathbb{Z}}\Gamma_x(\frac{\omega}{T}+\frac{2\pi k}{T})\left|S(\frac{\omega}{T}+\frac{2\pi k}{T})\right|^2}, \label{eq:stwienerunc}
\end{align}
with the reconstruction filter $W(\omega)=\Gamma_x(\omega)S(\omega)$.

In Section \ref{sec:samp_stationarity}, we will introduce the sampling framework for GWSS, which is a counterpart to what we describe in this section.

\section{Sampling of Graph Signals and Graph Wide Sense Stationarity}\label{sec:samp_stationarity}

In this section, we introduce the sampling and recovery framework used throughout the paper.
Then, two representative sampling methods for graph signals are reviewed.
We also formally define GWSS as an analog of WSS.

\subsection{Sampling and recovery framework for graph signals}
We begin by reviewing generalized graph signal sampling \cite{hara_generalized_2020,tanaka_generalized_2020}, illustrated in Fig. \ref{gene_corrupt_filter_res}. 
The main differences among various methods stem from the sampling domain and signal generation model.

Let $\bm{c}\in \mathbb{C}^K (K\leq N)$ be the sampled graph signal. When $\mb{S}^*$ is the sampling transformation, $\bm{c}=\mb{S}^*\bm{x}$, 
we would like to recover $\bm{x}$ from the noisy samples $\bm{y} = \bm{c}+\bm{\eta}$ using a reconstruction transformation $\mathbf{W} \in \mathbb{C}^{N\times K}$. The recovered signal is then
\begin{align}
\tilde{\bm{x}}=\mathbf{WH}\bm{y}=\mathbf{WH}(\bm{c}+\bm{\eta})=\mathbf{WH}(\mathbf{S}^*\bm{x}+\bm{\eta}),\label{recov}
\end{align}
where $\mathbf{H} \in \mathbb{C}^{K
\times K}$ is some transformation that compensates for non-ideal measurements. 
It is an analog of the generalized sampling framework introduced in Section \ref{sec:gen_sampling}: We use three filters $\mb{S}$, $\mb{H}$, and $\mb{W}$.
In the following, we seek the best $\mathbf{H}$ so that $\tilde{\bm{x}}$ is close to $\bm{x}$ in some sense under appropriate signal assumptions \cite{eldar_nonideal_2006,eldar_minimum_2006}.

\subsection{Sampling Methods}
Here, we introduce representative graph signal sampling operators in the vertex and graph frequency domains.

\subsubsection{Sampling in Vertex Domain}

Vertex domain sampling is expressed as the selection of samples on a vertex subset. Therefore, it  corresponds to nonuniform sampling in the time domain.

Sampling in the vertex domain is defined as follows:

\begin{definition}[Sampling in the vertex domain]\label{definition:vertexsamp}
Let $\bm{x}\in\mathbb{C}^N$ be the original graph signal and $\mb{G} \in\mathbb{C}^{N\times N}$ be an arbitrary graph filter. In addition, let $\mb{I}_\mathcal{M}\in \{0,1\}^{K\times N}$ be a submatrix of the identity matrix $\mb{I}_N$ extracting $K=|\mathcal{M}|$ rows corresponding to the sampling set $\mathcal{M}$. The sampled graph signal $\bm{c}\in\mathbb{C}^K$ is given as follows: 
\begin{align}
\bm{c}=\mb{I}_\mathcal{M}\mb{G}\bm{x}.\label{eq:vertexsamp}
\end{align}
\end{definition}
\noindent The sampling matrix is then expressed as $\mb{S}^*=\mb{I}_\mathcal{M}\mb{G}$.
In contrast to time domain signals, $\mb{I}_\mathcal{M}$ may depend on the graph since local connectivity of the graph is often irregular.
This implies that there may exist a ``best'' vertex set for graph signal sampling described in \cite{tanaka_sampling_2020}.

\subsubsection{Sampling in the GFT Domain}\label{subsection:gftdomainsamp}

When sampling SI signals, the spectrum folding phenomenon occurs \cite{eldar_sampling_2015}.
Graph frequency domain sampling utilizes the behavior in the graph spectrum.

Formally, graph frequency domain sampling, which is a counterpart of \eqref{eq:std_freq_samp}, is defined as follows:

\begin{definition}[Sampling in the graph frequency domain \cite{tanaka_spectral_2018}]\label{definition:freqsamp}
Let $\hat{\bm{x}} \in \mathbb{C}^N$ be the original signal in the graph frequency domain, i.e., $\hat{\bm{x}} = \mathbf{U}^* \bm{x}$, and let $S(\lambda_i)$ be an arbitrary sampling filter defined in the graph frequency domain. For any sampling ratio $M\in \mathbb{Z}$\footnote{$M$ is assumed to be a divisor of $N$ for simplicity.}, the sampled graph signal in the graph frequency domain is given by $\hat{\bm{c}} \in \mathbb{C}^{K}$, where $K=N/M$, and
\begin{align}
\hat{c}(\lambda_i)=\sum_{l=0}^{M-1}S(\lambda_{i+Kl}) \hat{x}(\lambda_{i+Kl}).\label{sampmod}
\end{align}
\end{definition}
\noindent
In matrix form, the sampled graph signal is represented as $\hat{\bm{c}}=\mathbf{D}_{\text{samp}}S(\bm{\Lambda})\hat{\bm{x}}$, where $\mathbf{D}_{\text{samp}}=[\mathbf{I}_K\ \mathbf{I}_K\ \cdots]\in \mathbb{C}^{K\times N }$. Then, we define the sampling matrix $\mathbf{S}^*$ as \cite{tanaka_generalized_2020}
\begin{equation}
\mathbf{S}^* = \mathbf{U}_{\text{reduced}}\mathbf{D}_{\text{samp}} S(\bm{\Lambda}) \mathbf{U}^*,\label{sampmat}
\end{equation}
where $\mathbf{U}_{\text{reduced}}\in \mathbb{C}^{K\times K}$ is an arbitrary unitary matrix, which may correspond to the GFT for a reduced-size graph. 
Theoretically, we can use any unitary matrix for $\mb{U}_{\text{reduced}}$.
One choice which has been studied in multiscale transforms for graph signals is graph reduction \cite{sakiyama_graph_2017}.

Here, we define reconstruction in the graph frequency domain, which is the counterpart of \eqref{eq:std_freq_reconst}, as follows:
\begin{align}
[\mb{U}^*\tilde{\bm{x}}](\lambda_i)= \hat{d}(\lambda_{i\,\mathrm{mod}\, K})W(\lambda_i) ,\label{reconstrep}
\end{align}
where $\hat{\bm{d}}\in\mathbb{C}^K$ is a vector composed of expansion coefficients and $W(\lambda_i)$ is the reconstruction filter kernel defined in the graph frequency domain.
Correspondingly, the reconstruction matrix is represented as
\begin{align}
    \mathbf{W}=\mathbf{U}W(\bm{\Lambda})\mathbf{D}_{\text{samp}}^*\mathbf{U}_{\text{reduced}}^{*}.\label{reconmat}
\end{align}
The reconstruction in (\ref{reconstrep}) is performed by replicating the original spectrum in the same manner as standard signal processing \cite{tanaka_spectral_2018}.

\subsection{Graph Wide Sense Stationarity}

Several definitions of GWSS have been proposed \cite{girault_stationary_2015,perraudin_stationary_2017,segarra_stationary_2017}.
A definition of GWSS based on the WSS by shift (see Collorary \ref{cor:wss_shift} in Appendix \ref{sec:stationarity}) is shown in Appendix \ref{app:gwss_type1}.

In this paper, we define GWSS based on the WSS by modulation (see Collorary \ref{cor:wss_localize} in Appendix \ref{sec:stationarity}) as follows:

\begin{definition}[Graph Wide Sense Stationary by Modulation (GWSS$_\text{M}$)]\label{definition:gwss2}
Let $\bm{x}$ be a graph signal on a graph $\mathcal{G}$\footnote{We suppose that $\mathcal{G}$ is connected for simplicity. Nevertheless, we can extend our definition of GWSS to the case where $\mathcal{G}$ is not connected.}. Then, $\bm{x}$ is a graph wide sense stationary process if and only if the following two conditions are satisfied for all $m$:
\begin{enumerate}
\itemequation[cond1:gwss2]{}{\mathbb{E}\left[(\bm{x}\star \bm{\delta}_n)_m\right]=\mu_x=\mathrm{const},}
\itemequation[cond2:gwss2]{}{\mathbb{E}\left[(\bm{x}\star \bm{\delta}_n-\mu_x\bm{1})_m(\bm{x}\star \bm{\delta}_{k}-\mu_x\bm{1})^*_m\right]=[\mathbf{\Gamma}_x]_{n,k}.} 
\end{enumerate}
The operator $\cdot \star \bm{\delta}_k$ is defined as follows: 
\begin{align}
\bm{x} \star \bm{\delta}_k\coloneqq \mb{M}\mathrm{diag}(u_{0}[k],u_{1}[k],\ldots)\mb{U}^*\bm{x},\label{eq:modulationoprator}
\end{align}
where $\mb{M}=[e^{j2\pi\cdot0/N}\bm{1},\,e^{j2\pi\cdot1/N}\bm{1},\,\cdots]$.
\end{definition}
\noindent
Note that $\cdot\star\bm{\delta}_k$ corresponds to the standard sinusoidal modulation, since $(\bm{x}\star\bm{\delta}_k)_n=(\mb{U}\exp(j\mb{\Omega})\mb{U}^*\bm{x})_k$ where $\mb{\Omega}=\text{diag}(0,2\pi/N,\ldots,2\pi (N-1)/N)$. Therefore, we refer to $\cdot\star\bm{\delta}_k$ as a modulation operator on a graph.
Derivation and motivation of Definition \ref{definition:gwss2} are given in Appendix \ref{app:gwss_type2}.

One of the important properties of Definition \ref{definition:gwss2} is that $\mb{\Gamma}_x$ is diagonalizable by $\mb{U}$ (cf. Appendix \ref{sec:stationarity_detail}). We refer to $\widehat{\Gamma}_x(\mb{\Lambda})\coloneqq \mb{U}^*\mb{\Gamma}_x\mb{U}$ as the graph PSD. 
This property is preferred in GSP since it parallels the Wiener-Khinchin relation in the time domain\cite{papoulis_probability_2002} and it is advantageous for exploiting spectral tools in GSP.

Recall that the autocovariance function is invariant to shift in the standard WSS \cite{papoulis_probability_2002}.
Since a graph operator is often used as a shift operator in the graph setting, the invariance with shift for GWSS can be translated to
\begin{align}
\mb{\Gamma}_x \mb{L}= \mb{L} \mb{\Gamma}_x.\label{eq:gwss_primal_prop}
\end{align}
This property is also used in \cite{jian_wide-sense_2021,hasanzadeh_piecewise_2019}.
Note that $\mb{\Gamma}_x$ is diagonalizable by $\mb{U}$ if and only if \eqref{eq:gwss_primal_prop} is satisfied \cite{stein_real_2005}. 
In fact, the alternative definitions of GWSS require additional assumptions to satisfy \eqref{eq:gwss_primal_prop}.
For example, \cite{girault_stationary_2015} requires that all eigenvalues of $\mb{L}$ are distinct. 
However, we sometimes encounter graphs whose graph operator has an eigenvalue with multiplicity greater than one \cite{perraudin_stationary_2017}.
Moreover, \cite{perraudin_stationary_2017,segarra_stationary_2017} require that the covariance $\mb{\Gamma}_x$ is expressed by a  polynomial in $\mb{L}$, while this assumption may not be true in general, including the multiple eigenvalue case. 
On the other hand, our GWSS definition does not require such assumptions, but it still satisfies \eqref{eq:gwss_primal_prop}.
As a result, Definition \ref{definition:gwss2} can be regarded as an extension of the existing GWSS.
The relationship among the GWSS definitions is discussed in detail in Appendix \ref{app:relation_gwss}.
Note that our generalized sampling is applicable under different definitions of GWSS with slight modifications.

\section{Graph Wiener filter: recovery for stochastic graph signals}
\label{sec:graph_wiener}

We now consider the design of the correction filter.
It is optimized such that the reconstructed graph signal is close to the original one in the mean squared sense.
We show that the resulting Wiener filter parallels that in the SI setting when graph spectral sampling (Definition \ref{definition:freqsamp}) is performed.

\subsection{Graph Wiener Filter}\label{sec:graph_wiener_vertex}
Suppose that $\bm{x}$ is a zero-mean GWSS process with a known covariance $\mb{\Gamma}_x$ and $\bm{\eta}$ is a zero-mean GWSS noise process with a known covariance $\mb{\Gamma}_\eta$, independent of each other. 
For simplicity, we suppose that columns of $\mb{W}$ and $\mb{S}$ satisfy the Riesz condition \cite{eldar_sampling_2015}, i.e., $\mb{W}^*\mb{W}$ and $\mb{S}^*\mb{S}$ are invertible.

We now formulate the design the correction filter. 
Unlike the formulation \eqref{eq:worst-case} in the time domain, we can directly minimize the expectation of the normed error in the graph setting
because the subspace of graph signals is finite-dimensional.

We consider minimizing the MSE between the original and reconstructed graph signals by solving
\begin{align}
\min_{\mathbf{H}}\  \mathbb{E}[\,\| \tilde{\bm{x}}-\bm{x}\|^2\,].\label{estimationerrorprob}
\end{align}
The optimal correction filter is obtained in the following Theorem.

\begin{theorem}\label{theorem:h_con}
Suppose that the input signal $\bm{x}$ and noise $\bm{w}$ are zero-mean GWSS processes with covariances $\mb{\Gamma}_x$ and $\mb{\Gamma}_w$, respectively, which are uncorrelated with each other.
Then, the solution of \eqref{estimationerrorprob} is given by
\begin{align}
\mathbf{H}_{\mathrm{PRE}}=(\mathbf{W}^*\mathbf{W})^{-1}\mathbf{W}^*\bm{\Gamma}_x \mathbf{S}(\mathbf{S}^*\bm{\Gamma}_x \mathbf{S}+\bm{\Gamma}_\eta)^{-1}.\label{wienervertex}
\end{align}
\end{theorem}
\begin{proof}
The MSE $\varepsilon_{\text{MSE}}=\mathbb{E}[\|\tilde{\bm{x}}-\bm{x}\|^2]$ is given by
\begin{align}
\varepsilon_{\text{MSE}}=&
\mathbb{E}[\tr \left(\mb{WHS}^*\bm{x}\bm{x}^*\mb{S}\mb{H}^*\mb{W}^*\right)]\nonumber\\
&+\mathbb{E}[\tr\left(-\bm{x}\bm{x}^*\mb{S}\mb{H}^*\mb{W}^*-\bm{x}\bm{x}^*\mb{WHS}^*+\bm{x}\bm{x}^*\right)]\nonumber\\
&+\mathbb{E}[\tr \left(\mb{WH}\bm{\eta}\bm{\eta}^*\mb{H}^*\mb{W}^*-\bm{x}\bm{\eta}^*\mb{H}^*\mb{W}^*\right)]\nonumber\\
& +\mathbb{E}[\tr\left(-\mb{WH}\bm{\eta}\bm{x}^*+\bm{\eta}\bm{\eta}^*\right)]\nonumber\\
=&\tr \left(\mb{WHS}^*\mb{\Gamma}_x\mb{S}\mb{H}^*\mb{W}^*-\mb{\Gamma}_x\mb{S}\mb{H}^*\mb{W}^*\right) \nonumber\\
&+\tr\left(-\mb{\Gamma}_x\mb{WHS}^*+\mb{\Gamma}_x\right)\nonumber\\
&+\tr \left(\mb{WH}\mb{\Gamma}_\eta\mb{H}^*\mb{W}^*+\mb{\Gamma}_\eta\right)\nonumber\\
=&\tr \left(\mb{S}^*\mb{\Gamma}_x\mb{S}\mb{H}^*\mb{W}^*\mb{WH}-2\mb{S}^*\mb{\Gamma}_x\mb{WH}+\mb{\Gamma}_x\right)\nonumber\\
& +\tr \left(\mb{\Gamma}_\eta\mb{H}^*\mb{W}^*\mb{W}\mb{H}+\mb{\Gamma}_\eta\right), \label{proof:eq1}
\end{align}
where the second equality holds since expectation and trace are interchangeable and the third equality follows by the cyclic property of trace. 

The infinitesimal perturbation of $\varepsilon_{\text{MSE}}$ with respect to $\mb{H}$ results in
\begin{align}
\mathrm{d}\varepsilon_{\text{MSE}}=&\tr\left( 2\mb{S}^*\mb{\Gamma}_x\mb{S}\mb{H}^*\mb{W}^*\mb{W} \mathrm{d}\mb{H}\right)\nonumber\\
&+\tr\left( -2\mb{S}^*\mb{\Gamma}_x\mb{W}\mathrm{d}\mb{H}+2\mb{\Gamma}_\eta\mb{H}^*\mb{W}^*\mb{W}\mathrm{d}\mb{H}\right).\label{proof:infinitestimal_e}
\end{align}
Here, $\mathrm{d}\varepsilon_{\text{MSE}}$ is defined by
\begin{align}
\mathrm{d}\varepsilon_{\text{MSE}}=\sum_{n,m} \left[\frac{\partial \varepsilon_{\text{MSE}}}{\partial \mb{H}}\circ\mathrm{d}\mb{H}\right]_{n,m}=\tr\left(\left(\frac{\partial \varepsilon_{\text{MSE}}}{\partial \mb{H}}\right)^*\mathrm{d}\mb{H}\right).\label{proof:infinitesimal}
\end{align}
As a result, the derivative of $\varepsilon_{\text{MSE}}$ in \eqref{proof:eq1} with respect to $\mb{H}$ results in
\begin{align}
\frac{\partial\varepsilon_{\text{MSE}}}{\partial \mb{H}}=\mathbf{W}^*\mathbf{W}\mathbf{H}(\mathbf{S}^*\bm{\Gamma}_x\mathbf{S}+\bm{\Gamma}_\eta)-\mathbf{W}^*\bm{\Gamma}_x\mathbf{S}.\label{proof:estimationerrorderiv}
\end{align}

Setting \eqref{proof:estimationerrorderiv} to zero, we obtain the optimal $\mb{H}$ of \eqref{wienervertex}.
\end{proof}

In Theorem \ref{theorem:h_con}, we consider the predefined case.
Next, we move on to the unconstrained case: The reconstruction filter $\mathbf{W}$ can also be freely chosen along with $\mathbf{H}$.
The optimal solution is obtained by solving
\begin{align}
\min_{\mb{W},\,\mb{H}} \mathbb{E}\left[\,\|\tilde{\bm{x}}-\bm{x}\|^2\,\right].\label{uncprob}
\end{align}
The optimal correction filter is given in the following Theorem.

\begin{theorem}\label{theorem:h_unc}
Suppose that the input signal $\bm{x}$ and noise $\bm{w}$ are zero-mean GWSS processes with covariances $\mb{\Gamma}_x$ and $\mb{\Gamma}_w$, respectively, which are uncorrelated with each other.
Then, the solution of \eqref{uncprob} is given by 
\begin{align}
\mathbf{H}_{\mathrm{UNC}}=(\mathbf{S}^*\bm{\Gamma}_x \mathbf{S}+\bm{\Gamma}_\eta)^{-1},\quad \mb{W}_{\mathrm{UNC}}=\mb{\Gamma}_x \mb{S}.\label{unconstwienervertex}
\end{align}
\end{theorem}

\begin{proof}
Since \eqref{wienervertex} is optimal for an arbitrary $\mb{W}$, we plug it into \eqref{proof:eq1} and have
\begin{align}
\varepsilon_{\text{MSE}}
=&\tr(\mb{P}_\mathcal{W}\mb{Z}\mb{S}^*\mb{\Gamma}_x\mb{S}\mb{Z}^*\mb{P}_\mathcal{W}-2\mb{\Gamma}_x\mb{S}\mb{Z}^*\mb{P}_\mathcal{W}+\mb{\Gamma}_x)\nonumber\\
& +\tr(\mb{P}_\mathcal{W}\mb{Z}\mb{\Gamma}_\eta\mb{Z}^*\mb{P}_\mathcal{W}+\mb{\Gamma}_w),\label{proof2:eq1}
\end{align}
where $\mb{P}_\mathcal{W}=\mb{W}(\mb{W}^*\mb{W})^{-1}\mb{W}^*$ and $\mb{Z}=\mb{\Gamma}_x\mb{S}(\mb{S}^*\mb{\Gamma}_x\mb{S}+\mb{\Gamma}_\eta)^{-1}$. In the same manner as \eqref{proof:estimationerrorderiv}, the derivative of $\varepsilon_{\text{MSE}}$ with respect to $\mb{W}$ is derived from \eqref{proof2:eq1} as
\begin{align}
\frac{\partial \varepsilon_{\text{MSE}}}{\partial \mb{W}}=&2\mb{P}_\mathcal{W}\mb{Z}\mb{S}^*\mb{\Gamma}_x\mb{S}\mb{Z}^*\mb{W}(\mb{W}^*\mb{W})^{-1}\nonumber\\
&-2\mb{\Gamma}_x\mb{S}\mb{Z}^*\mb{W}(\mb{W}^*\mb{W})^{-1}\nonumber\\
&+2\mb{P}_\mathcal{W}\mb{Z}\mb{\Gamma}_\eta\mb{Z}^*\mb{W}(\mb{W}^*\mb{W})^{-1}.\label{proof2:deriv}
\end{align}
By setting \eqref{proof2:deriv} to zero, we have
\begin{align}
(\mb{P}_\mathcal{W}\mb{\Gamma}_x\mb{S}-\mb{\Gamma}_x\mb{S})\mb{Z}^*\mb{W}(\mb{W}^*\mb{W})^{-1}&=\bm{0}.
\end{align}
Therefore,  $\mb{P}_\mathcal{W}$ is necessary to be the identity mapping over $\mathcal{R}(\mb{\Gamma}_x\mb{S})$, leading to $\mathcal{W}=\mathcal{R}(\mb{\Gamma}_x\mb{S})$ where $\mathcal{R}(\cdot)$ is the range space of a matrix. Graph signal recovery with \eqref{wienervertex} can be expressed as
\begin{align}
\mb{WH}=\mb{P}_\mathcal{W}\mb{\Gamma}_x\mb{S}(\mb{S}^*\mb{\Gamma}_x\mb{S}+\mb{\Gamma}_\eta)^{-1}.\label{proof:unc_sol_const}
\end{align}
Since $\mathcal{W}= \mathcal{R}(\mb{\Gamma}_x\mb{S})$, we can choose $\mb{W}=\mb{\Gamma}_x\mb{S}$ and obtain $\mb{H}=(\mb{S}^*\mb{\Gamma}_x\mb{S}+\mb{\Gamma}_\eta)^{-1}$. This completes the proof.
\end{proof}

So far, we considered a general solution for stochastic recovery. 
By choosing graph frequency domain sampling, we can show that the resulting recovery parallels that in SI space, studied in \cite{eldar_sampling_2015,eldar_nonideal_2006}.

\subsection{Special cases for GFT domain sampling}\label{sec:graph_wiener_gft}

Suppose that the graph signal and noise conform to zero-mean GWSS processes with PSD $\widehat{\Gamma}_x(\lambda)$ and $\widehat{\Gamma}_\eta(\lambda)$, respectively. 
We assume that $\mathbf{S}^*$ and $\mathbf{W}$ are defined in the graph frequency domain by \eqref{sampmat} and \eqref{reconmat}, respectively.
In this setting, \eqref{wienervertex} is diagonalizable by $\mathbf{U}_{\text{reduced}}$, i.e., it has a  graph frequency response
\begin{align}
H_{\text{PRE}}(\lambda_i)=\frac{\sum_l \Gamma_x(\lambda_{i+Kl})W^*(\lambda_{i+Kl})S(\lambda_{i+Kl})}{R_{W}(\lambda_i)(\sum_l\widehat{\Gamma}_x(\lambda_{i+Kl})|S(\lambda_{i+Kl})|^2+\widehat{\Gamma}_\eta(\lambda_{i}))},\label{wienerfreq}
\end{align}
where
\begin{align}
R_{W}(\lambda_i)\coloneqq\sum_{l=0}^{M-1}W^*(\lambda_{i+Kl})W(\lambda_{i+Kl}).
\end{align}
Similarly, the  graph spectral responses of (\ref{unconstwienervertex}) are given by
\begin{align}
    H_{\text{UNC}}(\lambda_i)&=\frac{1}{\sum_l \Gamma_x(\lambda_{i+Kl})|S(\lambda_{i+Kl})|^2+\Gamma_\eta(\lambda_{i})},\label{unconstwienerfreq}
\end{align}
with reconstruction filter $W(\lambda_i)= \Gamma_x(\lambda_i)S(\lambda_i)$. The correction filters \eqref{wienerfreq} and \eqref{unconstwienerfreq} inherit frequency responses of the Wiener filter in standard sampling \eqref{eq:stdwienercons} and \eqref{eq:stwienerunc}, respectively.

In the following, we discuss the relationship among the proposed graph Wiener filter and existing generalized graph signal sampling. 

\section{Relationship to other priors}\label{sec:relation_priors}

In this section, we consider the relationship between the proposed graph Wiener filter and graph signal recovery under subspace and smoothness priors \cite{tanaka_generalized_2020}.
For simplicity, we only consider the predefined case.
The unconstrained case can be derived in a similar fashion.

\subsection{Subspace prior}

Under the subspace prior, we assume the following graph signal model \cite{tanaka_generalized_2020}.
\begin{equation}
    \bm{x}=\mb{A}\bm{d},
\end{equation}
where $\mb{A}\in\mathbb{C}^{N\times K}$ is a known generator and $\bm{d}\in\mathbb{C}^{K}$ is expansion coefficients.

We consider a relationship between the subspace prior and the stochastic one. 
Suppose that $\bm{d}$ is a random vector and
$\mb{\Sigma}_d=\mathbb{E}[\bm{d}\bm{d}^*]$ is the covariance of $\bm{d}$.
The covariance of $\bm{x}$ is then written as
\begin{align}
\mb{\Sigma}_x=&\mathbb{E}[\mb{A}\bm{d}\bm{d}^*\mb{A}^*]\nonumber\\
=&\mb{A}\mathbb{E}[\bm{d}\bm{d}^*]\mb{A}^*\nonumber\\
=&\mb{A}\mb{\Sigma}_d\mb{A}^*.\label{eq:cov_subspace}
\end{align}
By substituting \eqref{eq:cov_subspace} to \eqref{wienervertex}, we have the following correction filter:
\begin{align}
\mb{H}_\text{MX,SUB}=(\mb{W}^*\mb{W})^{-1}\mb{W}^*\mb{A}\mb{\Sigma}_d\mb{A}^*\mb{S}(\mb{S}^*\mb{A}\mb{\Sigma}_d\mb{A}^*\mb{S}+\mb{\Gamma}_\eta)^{-1}.\label{eq:wienervertex_subspace_0}
\end{align}

For subspace prior, 
if $\mb{S}^*\mb{A}$ is invertible and there is no noise,
\eqref{eq:wienervertex_subspace_0} reduces to
\begin{align}
\mb{H}_\text{MX,SUB}=&(\mb{W}^*\mb{W})^{-1}\mb{W}^*\mb{A}\mb{\Sigma}_d\mb{A}^*\mb{S}(\mb{S}^*\mb{A}\mb{\Sigma}_d\mb{A}^*\mb{S})^{-1}\nonumber\\
=&(\mb{W}^*\mb{W})^{-1}\mb{W}^*\mb{A}\mb{\Sigma}_d\mb{A}^*\mb{S}(\mb{A}^*\mb{S})^{-1}\mb{\Sigma}_d^{-1}(\mb{S}^*\mb{A})^{-1}\nonumber\\
=&(\mb{W}^*\mb{W})^{-1}\mb{W}^*\mb{A}(\mb{S}^*\mb{A})^{-1},\label{eq:wienervertex_subspace}
\end{align}
where we assume that $\mb{\Sigma}_d$ is invertible.
As observed, $\mb{H}_\text{MX,SUB}$ does not depend on $\mb{\Sigma}_d$.
Interestingly, the solution in \eqref{eq:wienervertex_subspace} coincides with the minimax solution for signal recovery under subspace prior \cite{tanaka_generalized_2020}.
As a result, we can view the minimax recovery under the subspace prior \eqref{eq:wienervertex_subspace} as a special case of \eqref{wienervertex} with random expansion coefficients and known $\mb{A}$.

\subsection{Smoothness prior}

We consider smoothness prior.
Under the smoothness prior, graph signals satisfy the following condition \cite{tanaka_generalized_2020}:
\begin{equation}
    \|\mb{V}\bm{x}\|^2\leq\rho^2,\label{eq:smo_prior}
\end{equation}
where $\mb{V}=\mb{U}V(\mb{\Lambda})\mb{U}^*$ is a smoothness measuring function, i.e., graph high-pass filter, and $\rho>0$ is a constant. We assume that $\mb{V}$ is bounded, i.e., $\mb{V}^*\mb{V}$ is invertible. 

To derive the solution of \eqref{estimationerrorprob} subject to the constraint in \eqref{eq:smo_prior}, the problem may be written as minimization of the Lagrangian
\begin{align}
L(\mb{H},\xi)=\mathbb{E}[\|\tilde{\bm{x}}-\bm{x}\|^2]+\xi(\mathbb{E}[\|\mb{V}\bm{x}\|^2]-\rho^2),\label{eq:lagrangian}
\end{align}
where $\xi$ is a constant.
The optimal solution of \eqref{eq:lagrangian} satisfies the following identity \cite{boyd_convex_2004}:
\begin{align}
\lambda(\mathbb{E}[\|\mb{V}\bm{x}\|^2]-\rho^2)=\xi(\|\mb{V}\mb{\Sigma}_x\mb{V}^*\|_F^2-\rho^2)=0.\label{eq:lagrangian_cond}
\end{align}
Note that the second equality of \eqref{eq:lagrangian_cond} is satisfied with either $\xi =0$ or $\|\mb{V}\mb{\Sigma}_x\mb{V}\|^2_F=\rho^2$.
Therefore, the nontrivial solution of \eqref{eq:lagrangian} can be obtained with $\|\mb{V}\mb{\Sigma}_x\mb{V}\|^2_F=\rho^2$. 
Since $\|\mb{V}(\mb{V}^*\mb{V})^{-1}\mb{V}^*\|_F^2=\|\mb{I}\|_F^2=N$, $\mb{\Sigma}_x$ can be written as
\begin{align}
\mb{\Sigma}_x=\frac{\rho^2}{N}(\mb{V}^*\mb{V})^{-1}.\label{eq:cov_smo}
\end{align}
Since the derivative of $L(\mb{H},\xi)$ with respect to $\mb{H}$ coincides with \eqref{proof:eq1}, we immediately obtain the following solution:
\begin{align}
\mb{H}_\text{MX,SMO}=&(\mb{W}^*\mb{W})^{-1}\mb{W}^*(\mb{V}^*\mb{V})^{-1}\mb{S}(\mb{S}^*(\mb{V}^*\mb{V})^{-1}\mb{S})^{-1},\label{eq:wienervertex_smoothness}
\end{align}
where we also assume the noiseless case ($\mb{\Gamma}_\eta = \mathbf{0}$ in \eqref{proof:eq1}) as in the subspace prior.

In fact, \eqref{eq:wienervertex_smoothness} coincides with the minimax solution under the smoothness prior \cite{tanaka_generalized_2020}. Similar to \eqref{eq:wienervertex_subspace}, we can view minimax recovery under the smoothness prior \eqref{eq:wienervertex_smoothness} as a special case of \eqref{wienervertex}.

\section{Signal Recovery Experiments}\label{sec:experiments}

\begin{table*}[t!]
\setlength{\tabcolsep}{0.13em}
\centering
\caption{Average MSEs of the reconstructed signals for the synthetic data on 20 independent realizations (in decibels).
We also show the standard deviations along with the MSEs.
ER, BL, FB, VS and SS denote an Erd\H{o}s-R\'{e}nyi graph, bandlimited, full-band, vertex domain sampling and spectral domain sampling, respectively. UNC, PRE, LS and MX denote unconstrained, predefined, least-squares and minimax criteria, respectively.}\label{tab:ex1}
\footnotesize
\begin{tabular}{c|c|c|c|cc|ccc|c|ccccc}
\hline
\multicolumn{4}{c|}{Prior} & \multicolumn{2}{c|}{Stochastic} & \multicolumn{3}{c|}{Smoothness \cite{tanaka_generalized_2020}} & \multirow{2}{*}{\begin{tabular}[c]{@{}c@{}}BL\\\cite{tanaka_spectral_2018}\end{tabular}} & \multirow{2}{*}{\begin{tabular}[c]{@{}c@{}}MKVV\\\cite{heimowitz_smooth_2018}\end{tabular}} & \multirow{2}{*}{\begin{tabular}[c]{@{}c@{}}MKVD\\\cite{narang_signal_2013}\end{tabular}} & \multirow{2}{*}{\begin{tabular}[c]{@{}c@{}}NLPD\\\cite{chen_discrete_2015}\end{tabular}} & \multirow{2}{*}{\begin{tabular}[c]{@{}c@{}}GSOD\\\cite{segarra_reconstruction_2016}\end{tabular}} & \multirow{2}{*}{\begin{tabular}[c]{@{}c@{}}NLPI\\\cite{narang_localized_2013}\end{tabular}}\\ \cline{1-9}
\multicolumn{4}{c|}{Reconst./ Criterion} & UNC & PRE & UNC & PRE LS & PRE MX&  &&&&&\\ \hline\hline
\multirow{8}{*}{Sensor} & \multirow{4}{*}{Noisy} & \multirow{2}{*}{BL} & VS & \textbf{-8.87}$\pm$1.18 & -8.86$\pm$1.19 & 443$\pm$14.4 & 208$\pm$11.5 & 442$\pm$14.4 & -8.69$\pm$1.16 & 136$\pm$23.9 & 68.9$\pm$30.4 & 119$\pm$23.6 & 18.9$\pm$0.10 & -6.79$\pm$1.15\\
 &  &  & SS & \textbf{-8.88}$\pm$1.99 &\textbf{-8.88}$\pm$1.99 & -8.13$\pm$1.99 & -6.99$\pm$1.98 & -8.63$\pm$1.99 & -8.19$\pm$1.99 & - & - & - & - & -\\ \cline{3-15} 
 &  & \multirow{2}{*}{FB} & VS & \textbf{-11.9}$\pm$1.15 & -10.7$\pm$1.15 & -4.73$\pm$1.44 & -9.82$\pm$1.21 & -9.59$\pm$1.21 &  - & 143$\pm$18.8 & 203$\pm$22.7 & 162$\pm$17.8 & 18.9$\pm$0.11 & -2.10$\pm$2.22\\
 &  &  & SS & \textbf{-14.1}$\pm$1.99 & -10.9$\pm$1.99 & -5.05$\pm$1.99 & -10.8$\pm$1.99 & -10.6$\pm$1.99 & - & - & - & - & - & -\\ \cline{2-15} 
 & \multirow{4}{*}{Clean} & \multirow{2}{*}{BL} & VS & \textbf{-8.99}$\pm$1.20 & -8.90$\pm$1.19 & 126$\pm$12.5 & 44.2$\pm$8.94 & 125$\pm$12.6 & -8.90$\pm$1.19 & 97.0$\pm$20.4 & -5.93$\pm$1.02 & -6.94$\pm$1.22 & 18.7$\pm$0.09 & -8.91$\pm$1.19\\
 &  &  & SS & \textbf{-8.99}$\pm$1.99 & -8.95$\pm$1.99 & \textbf{-8.99}$\pm$1.99 & -8.88$\pm$1.99 & -8.95$\pm$1.99 & -8.98$\pm$1.99 & - & - & - & - & -\\ \cline{3-15} 
 &  & \multirow{2}{*}{FB} & VS & \textbf{-13.5}$\pm$1.35 & -11.5$\pm$1.23 & -6.66$\pm$1.38 & -11.4$\pm$1.23 & -11.4$\pm$1.23 & -  & 132$\pm$19.7 & 176$\pm$27.6 & 152$\pm$17.1 & 18.8$\pm$0.11 & -4.01$\pm$2.13\\
 &  &  & SS & \textbf{-14.7}$\pm$1.99 & -11.2$\pm$1.99 & -5.51$\pm$1.99 & -11.1$\pm$1.99 & -11.0$\pm$1.99 & - & - & - & - & - & -\\ \hline\hline
\multirow{8}{*}{ER} & \multirow{4}{*}{Noisy} & \multirow{2}{*}{BL} & VS & \textbf{-18.6}$\pm$1.40 & -16.6$\pm$1.44 & 209$\pm$19.2 & 198$\pm$18.5 & 209$\pm$19.2 & -17.3$\pm$1.43 & -14.8$\pm$1.46 & 30.1$\pm$22.8 & 142$\pm$12.2 & 18.0$\pm$0.32 & -16.5$\pm$1.40 \\
 &  &  & SS & -24.2$\pm$1.99 & -18.2$\pm$1.99 & -24.2$\pm$1.99 & \textbf{-24.3}$\pm$1.99 & -12.6$\pm$1.99 & -18.2$\pm$1.99 & - & - & - & - & -\\ \cline{3-15} 
 &  & \multirow{2}{*}{FB} & VS & -\textbf{18.8}$\pm$1.53 & -17.2$\pm$1.46 & -14.3$\pm$1.50 & -14.1$\pm$1.43 & -14.1$\pm$1.43 & - & -14.8$\pm$1.46 & 53.1$\pm$14.1 & 147$\pm$11.8 & 18.1$\pm$0.31 & -16.8$\pm$1.42\\
 &  &  & SS & \textbf{-25.9}$\pm$2.00 & -19.5$\pm$2.00 & -21.3$\pm$1.99 & -19.5$\pm$1.99 & -19.5$\pm$1.99 & - & - & - & - & - & -\\ \cline{2-15} 
 & \multirow{4}{*}{Clean} & \multirow{2}{*}{BL} & VS & \textbf{-26.7}$\pm$1.64 & -17.5$\pm$1.51 & 62.4$\pm$15.9 & 66.1$\pm$14.0 & 62.1$\pm$15.9 & -17.6$\pm$1.47 & -14.8$\pm$1.46 & -4.95$\pm$22.8 & 137$\pm$12.7 & 17.8$\pm$0.33 & -16.9$\pm$1.46\\
 &  &  & SS & \textbf{-26.7}$\pm$2.00 & -18.7$\pm$1.99 & \textbf{-26.7}$\pm$2.00 & -13.9$\pm$2.00 & -18.7$\pm$1.99 & -26.5$\pm$2.00 & - & - & - & - & -\\ \cline{3-15} 
 &  & \multirow{2}{*}{FB} & VS & \textbf{-25.2}$\pm$1.91 & -18.8$\pm$1.61 & -19.9$\pm$1.63 & -18.8$\pm$1.61 & -18.8$\pm$1.61 & - & -14.8$\pm$1.46 & 14.4$\pm$20.2 & 141$\pm$12.2 & 18.0$\pm$0.32 & -17.7$\pm$1.54\\
 &  &  & SS & \textbf{-28.2}$\pm$2.01 & -20.1$\pm$2.00 & -22.1$\pm$2.00 & -20.1$\pm$2.00 & -20.1$\pm$2.00 & - & - & - & - & - & -\\ \hline\hline
\multirow{8}{*}{2D grid} & \multirow{4}{*}{Noisy} & \multirow{2}{*}{BL} & VS & \textbf{-8.25}$\pm$1.16 & -8.24$\pm$1.16 & 311$\pm$18.1 & 204$\pm$13.8 & 310$\pm$18.1 & -8.11$\pm$1.15 & -8.07$\pm$1.15 & -2.85$\pm$17.1 & 132$\pm$23.5 & 18.9$\pm$0.09 & -5.96$\pm$1.14\\
 &  &  & SS & \textbf{-8.31}$\pm$1.99 & -8.24$\pm$1.99 & -7.57$\pm$1.99 & -6.13$\pm$1.98 & -7.85$\pm$1.99 & -7.67$\pm$1.99 & - & - & - & - & -\\ \cline{3-15} 
 &  & \multirow{2}{*}{FB} & VS & \textbf{-12.4}$\pm$1.15 & -10.6$\pm$1.17 & -5.58$\pm$1.35 & -10.5$\pm$1.16 & -10.3$\pm$1.15 & - & -8.12$\pm$1.15 & 66.2$\pm$18.7 & 156$\pm$18.4 & 18.9$\pm$0.10 & -2.34$\pm$1.88 \\
 &  &  & SS & \textbf{-13.7}$\pm$1.99 & -10.7$\pm$1.99 & -5.75$\pm$1.99 & -10.6$\pm$1.99 & -10.4$\pm$1.99 & - & - & - & - & - & - \\ \cline{2-15} 
 & \multirow{4}{*}{Clean} & \multirow{2}{*}{BL} & VS & \textbf{-8.51}$\pm$1.18 & -8.33$\pm$1.18 & 72.0$\pm$17.3 & 36.0$\pm$10.4 & 71.2$\pm$17.2 & -8.32$\pm$1.17 & -8.23$\pm$1.16 & -6.38$\pm$1.08 & -6.41$\pm$1.62 & 18.7$\pm$0.11 &-8.32$\pm$1.17\\
 &  &  & SS & \textbf{-8.51}$\pm$1.99 & -8.39$\pm$1.99 & \textbf{-8.51}$\pm$1.99 & -8.25$\pm$1.99 & -8.39$\pm$1.99 & -8.48$\pm$1.99 & - & - & - & - & -\\ \cline{3-15} 
 &  & \multirow{2}{*}{FB} & VS & \textbf{-14.0}$\pm$1.29 & -11.3$\pm$1.24 & -6.45$\pm$1.37 & -11.2$\pm$1.24 & -11.2$\pm$1.24 & - & -8.13$\pm$1.15 & 61.3$\pm$20.6 & 153$\pm$18.3 & 18.9$\pm$0.11 & -3.29$\pm$1.82\\
 &  &  & SS & \textbf{-14.5}$\pm$2.00 & -11.0$\pm$1.99 & -6.06$\pm$1.99 & -10.9$\pm$1.99 & -10.8$\pm$1.99 & - & - & - & - & - & -\\ \hline
\end{tabular}
\end{table*}

In this section, we validate the effectiveness of the proposed method via signal recovery experiments. We demonstrate that the approach described in Section \ref{sec:graph_wiener} reduces the reconstruction error in comparison with existing recovery techniques under the stochastic setting.

\begin{figure}[t!]
  \centering
  \includegraphics[width=1\linewidth]{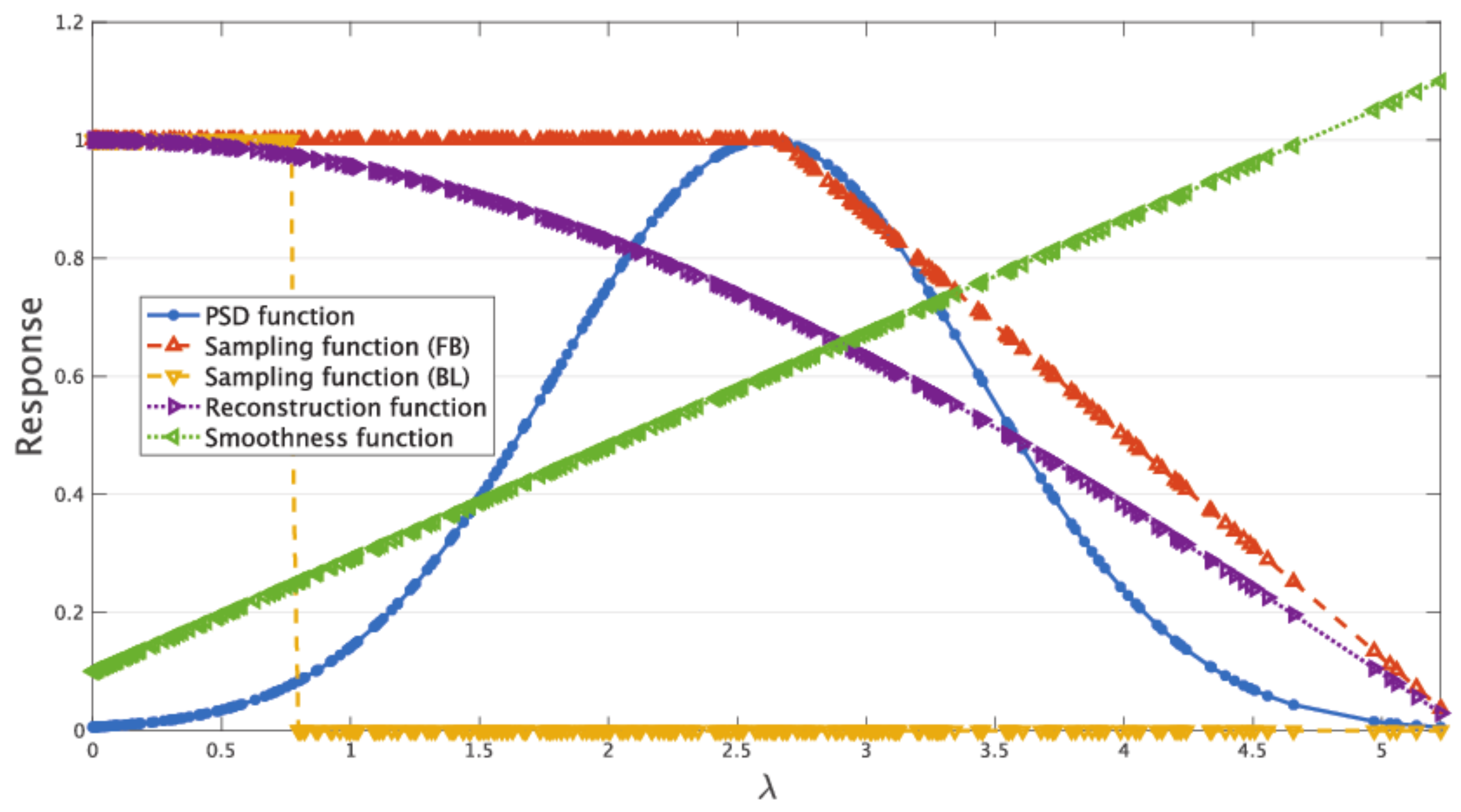}
  \caption{Graph frequency responses of several functions used for experiments.}\label{ex1_res_filter}
\end{figure}

\begin{figure*}[t!]
\centering
 \subfigure[][Original]
  {\includegraphics[width=0.23\linewidth]{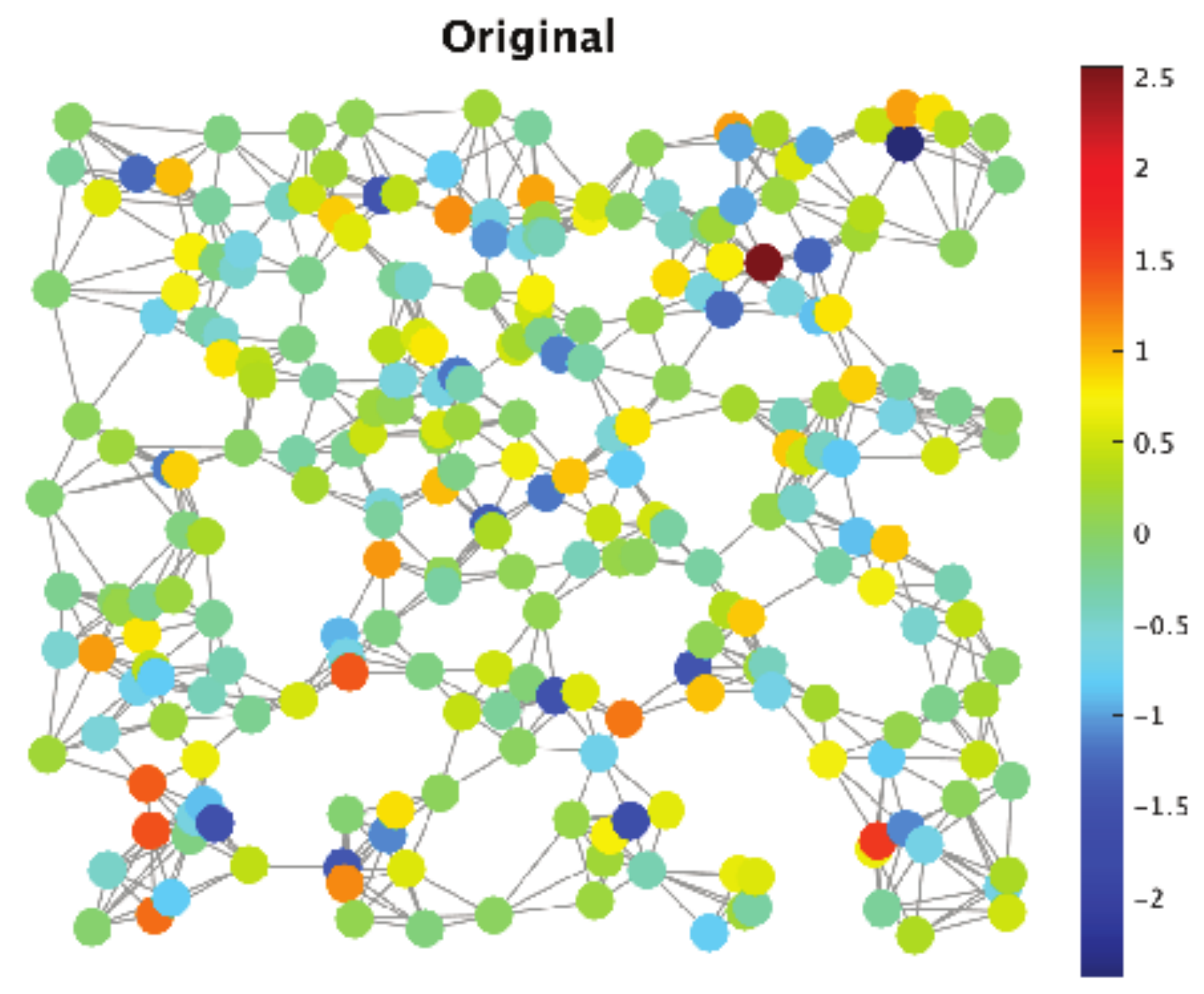} \label{ex1_ori}}
 \subfigure[][Proposed: UNC ST.]
  {\includegraphics[width=0.23\linewidth]{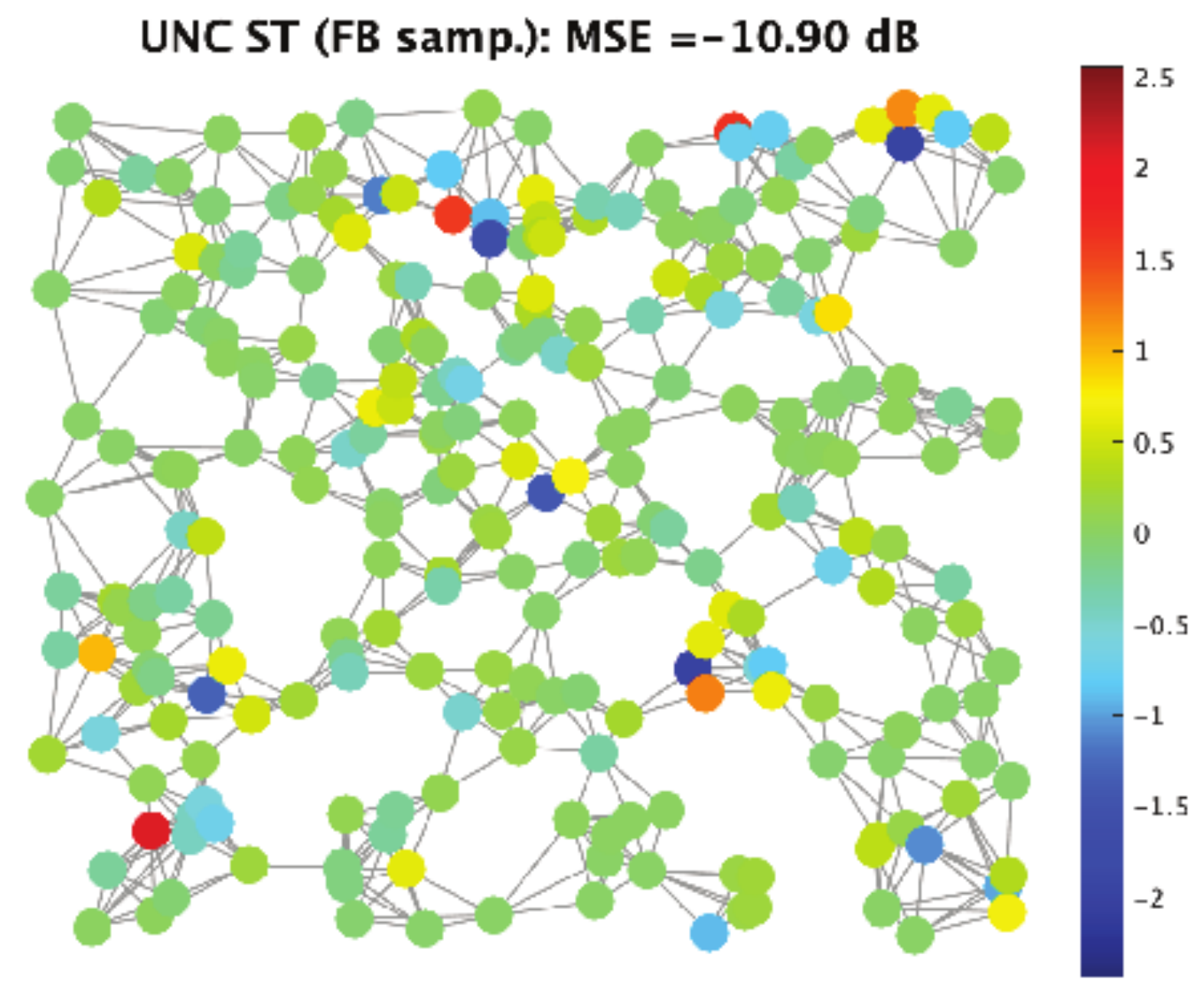} \label{ex1_unc_st_nbl}}
  \subfigure[][Proposed.: PRE ST.]
  {\includegraphics[width=0.23\linewidth]{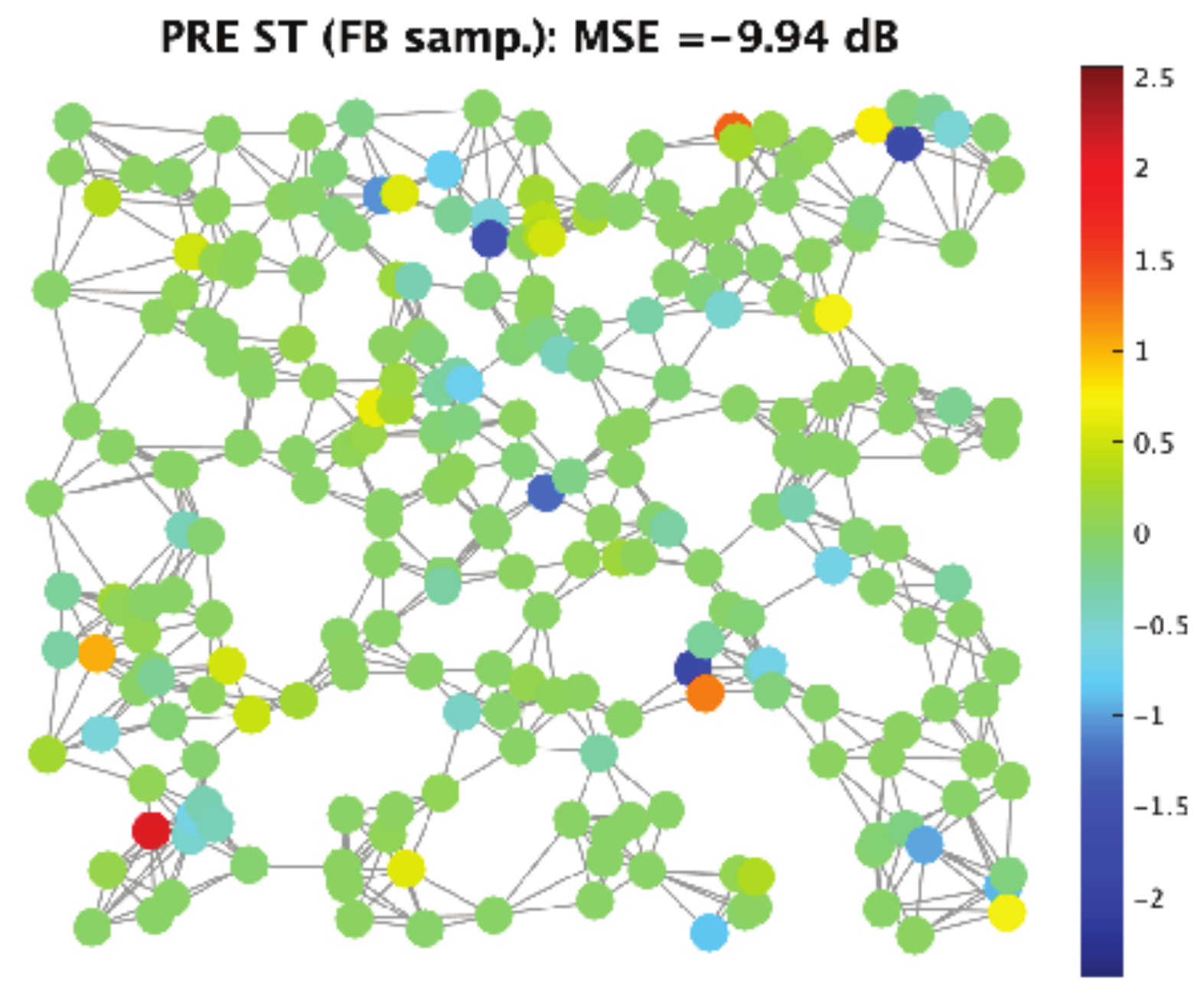} \label{ex1_con_st_nbl}}
 \subfigure[][UNC SM\cite{tanaka_generalized_2020}.]
  {\includegraphics[width=0.23\linewidth]{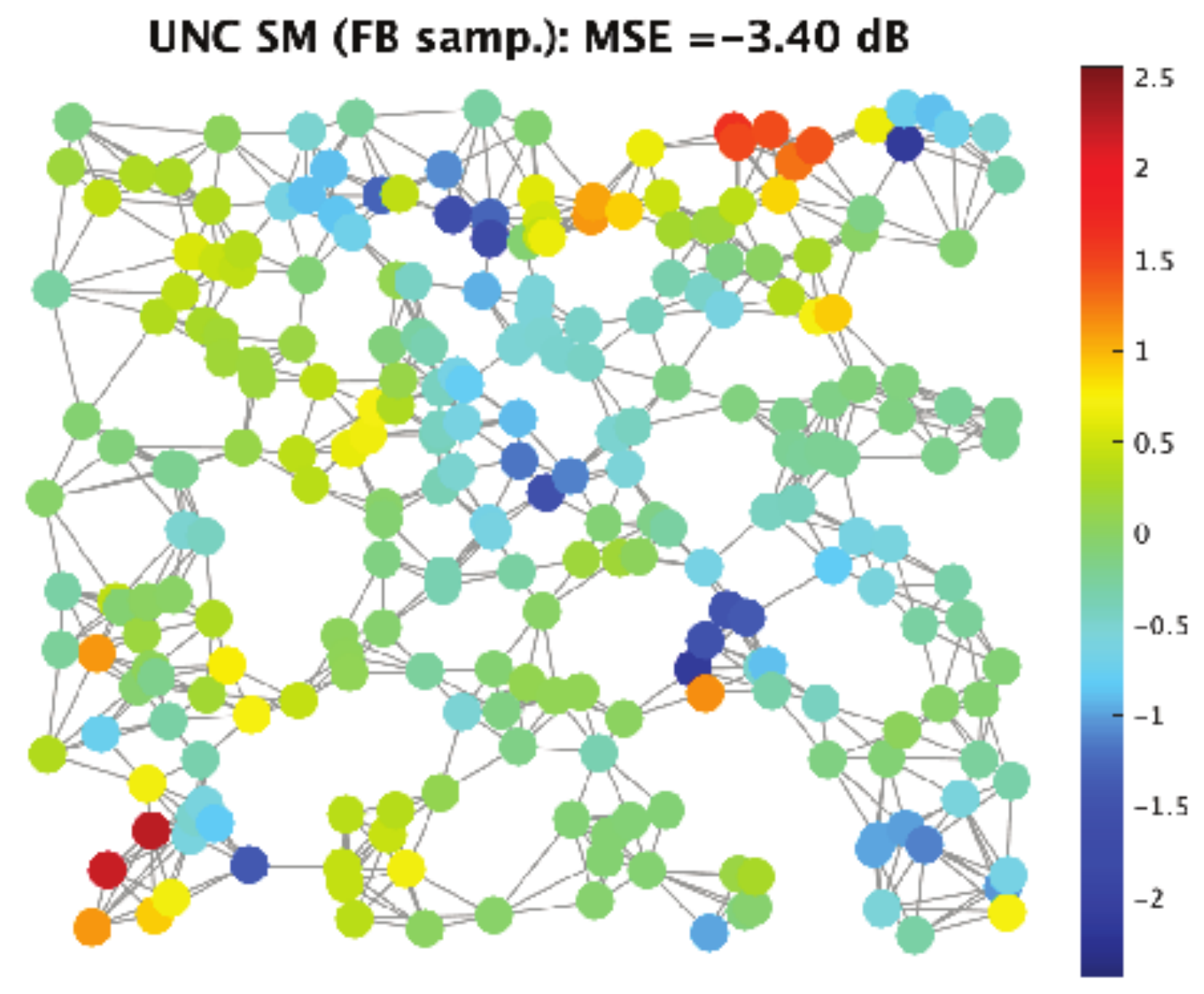} \label{ex1_unc_sm_nbl}}
 \subfigure[][PRE SM LS\cite{tanaka_generalized_2020}.]
  {\includegraphics[width=0.23\linewidth]{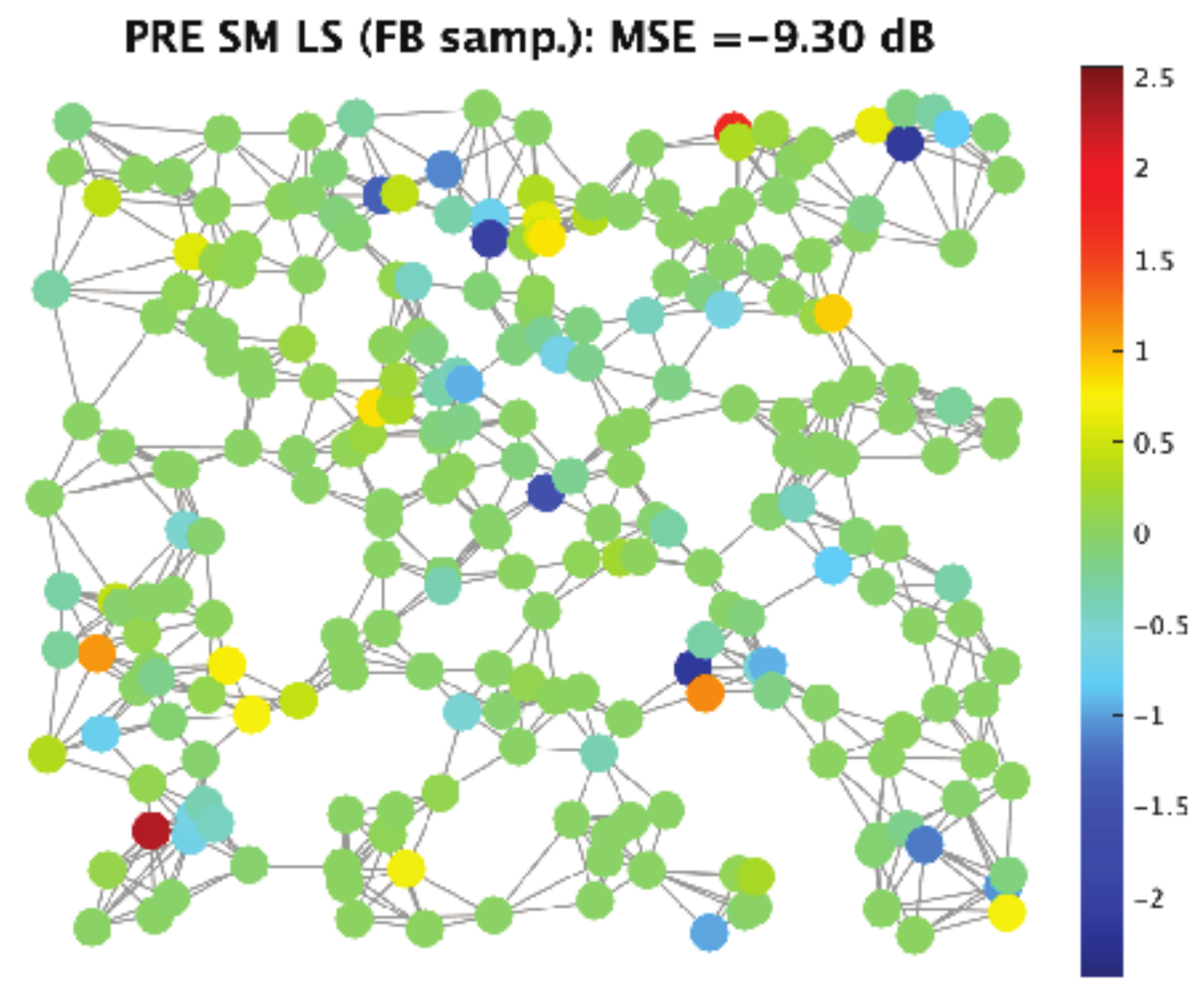} \label{ex1_con_sm_ls_nbl}}
 \subfigure[][PRE SM MX\cite{tanaka_generalized_2020}.]
  {\includegraphics[width=0.23\linewidth]{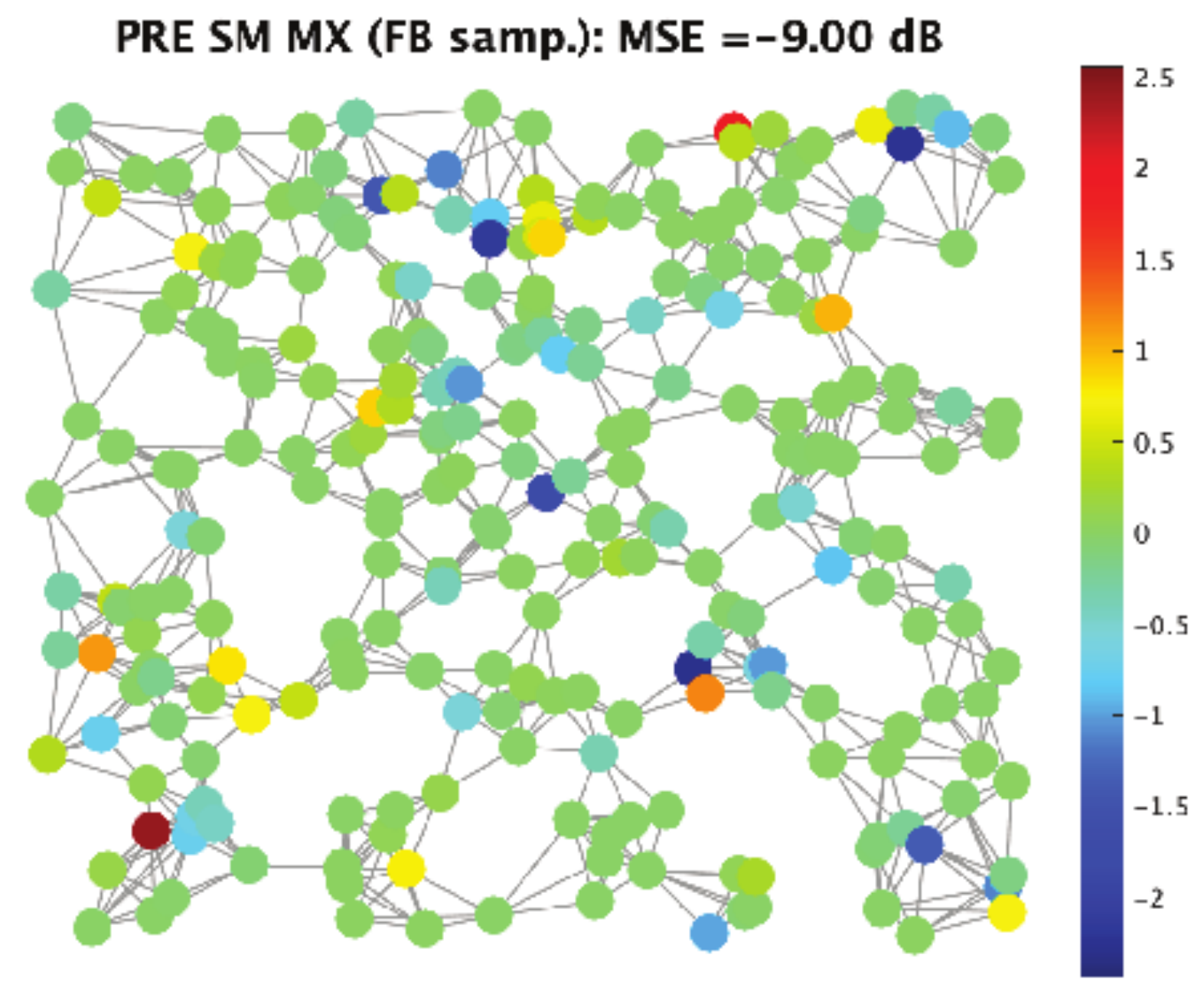} \label{ex1_con_sm_mx_nbl}}
 \subfigure[][BL\cite{tanaka_spectral_2018}.]
  {\includegraphics[width=0.23\linewidth]{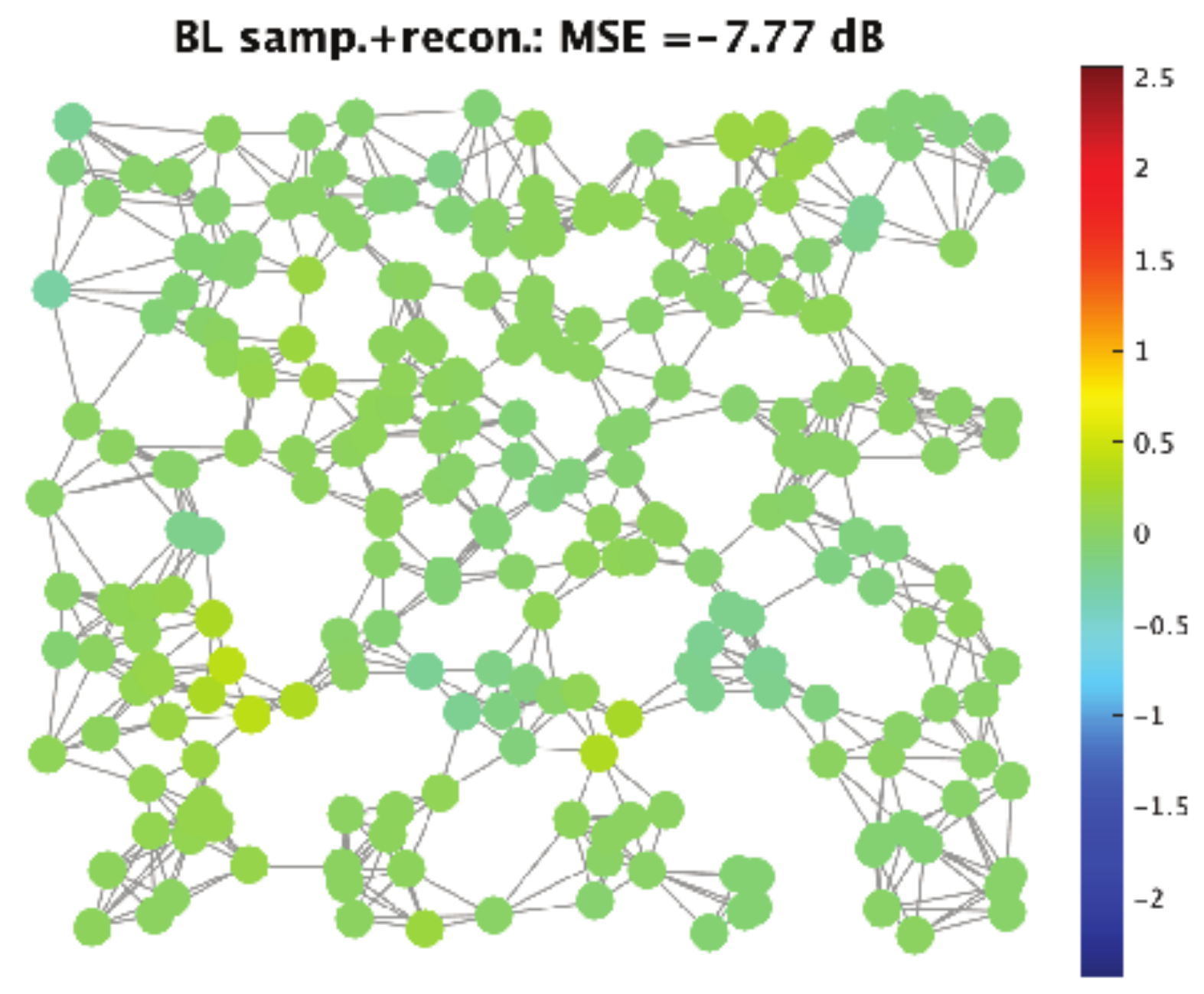} \label{ex1_bl}}
 \subfigure[][MKVV\cite{heimowitz_smooth_2018}.]
  {\includegraphics[width=0.23\linewidth]{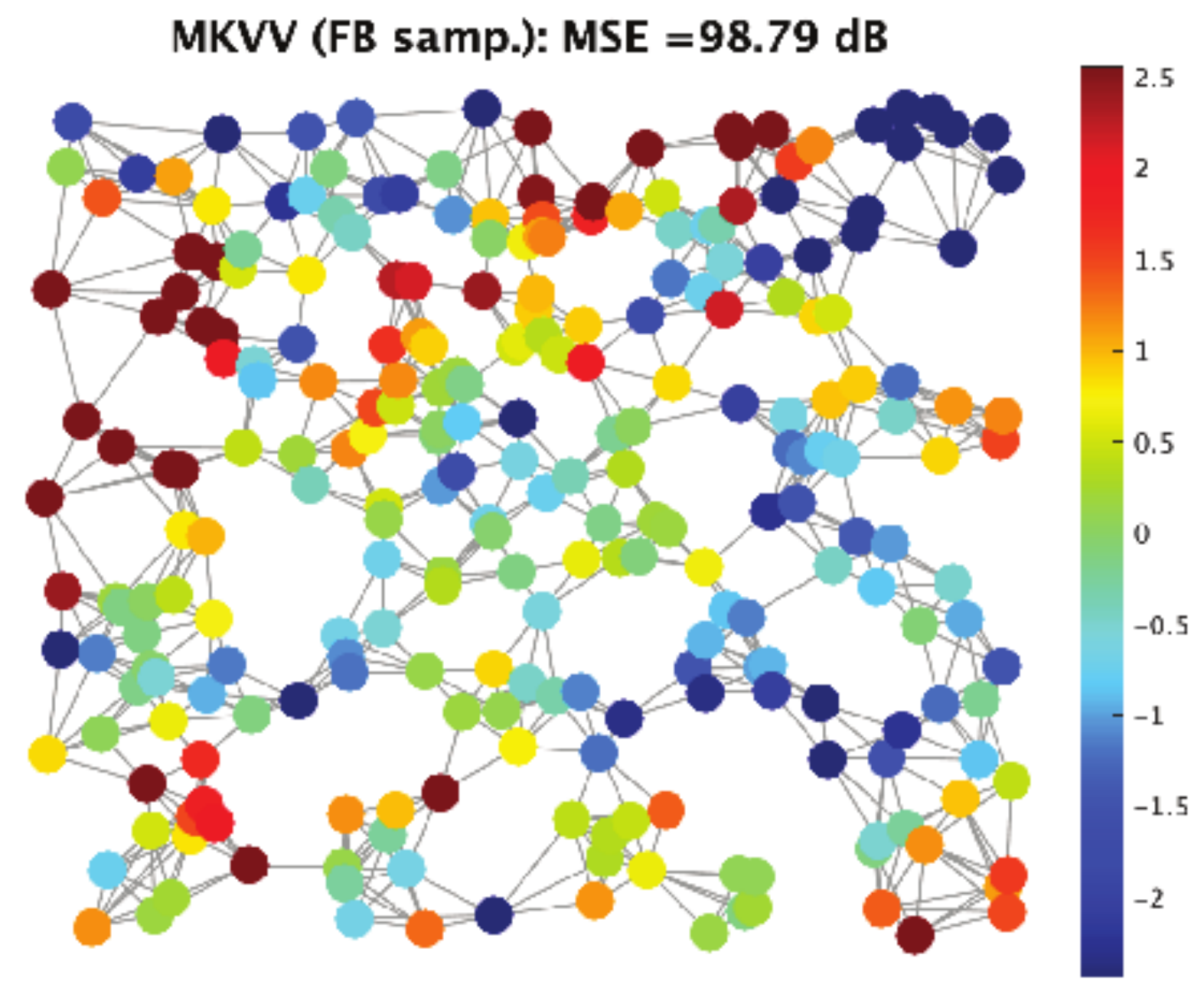} \label{ex1_unc_st_bl}}
 \subfigure[][MKVD\cite{narang_signal_2013}.]
  {\includegraphics[width=0.23\linewidth]{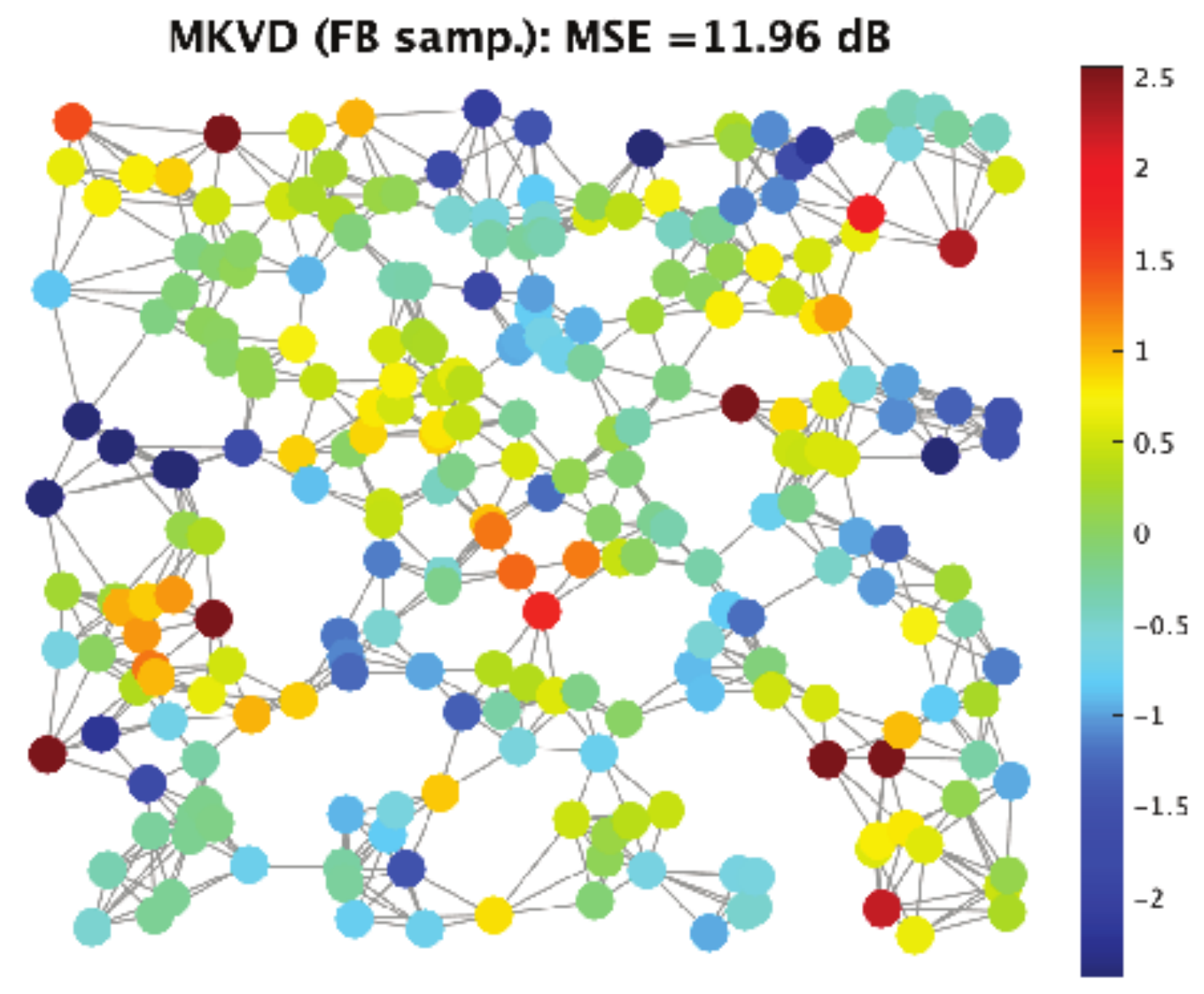} \label{ex1_con_st_bl}}
 \subfigure[][NLPD\cite{chen_discrete_2015}.]
  {\includegraphics[width=0.23\linewidth]{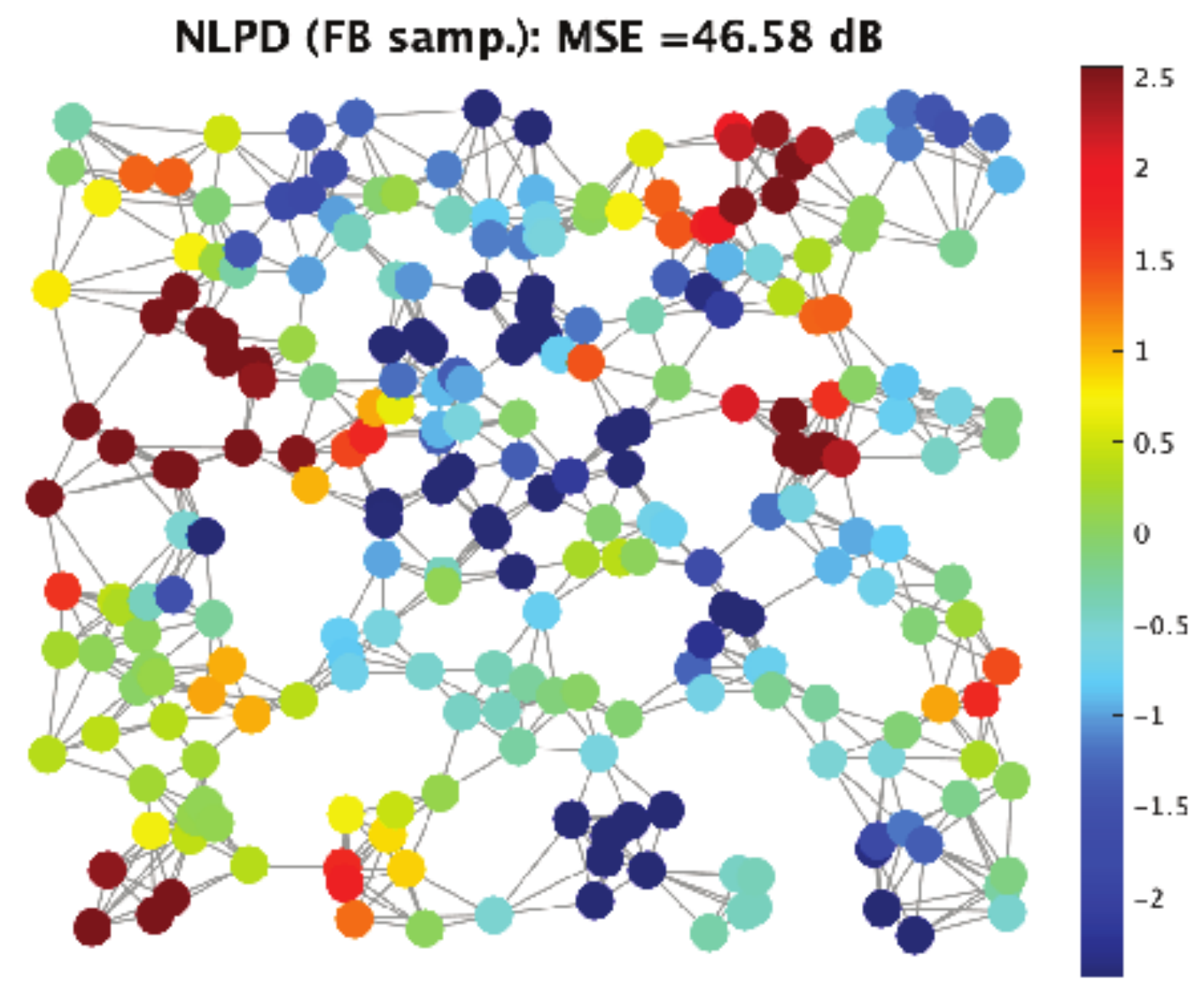} \label{ex1_unc_sm_bl}}
 \subfigure[][GSOD\cite{segarra_reconstruction_2016}.]
  {\includegraphics[width=0.23\linewidth]{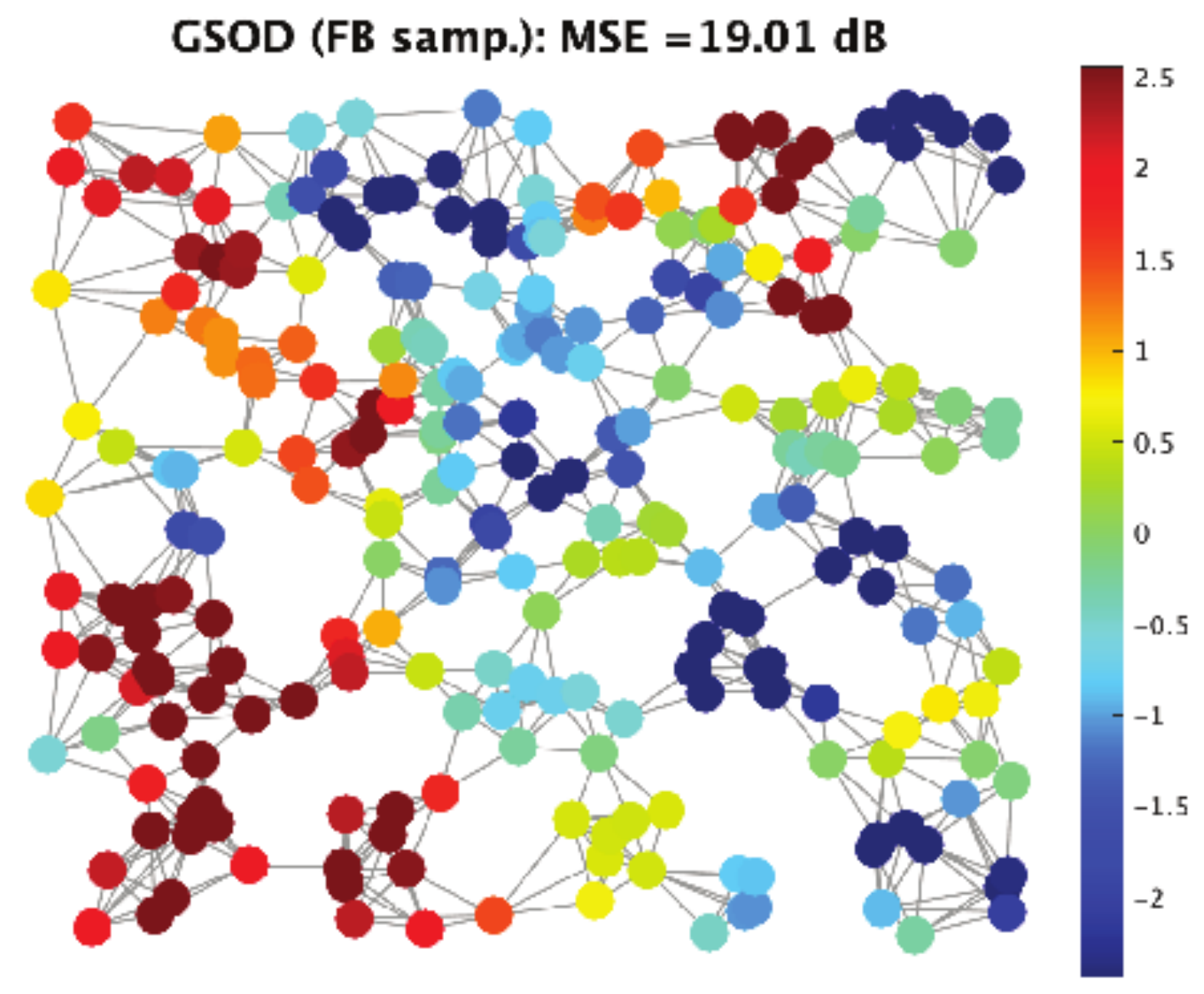} \label{ex1_con_sm_ls_bl}}
 \subfigure[][NLPI\cite{narang_localized_2013}. ]
  {\includegraphics[width=0.23\linewidth]{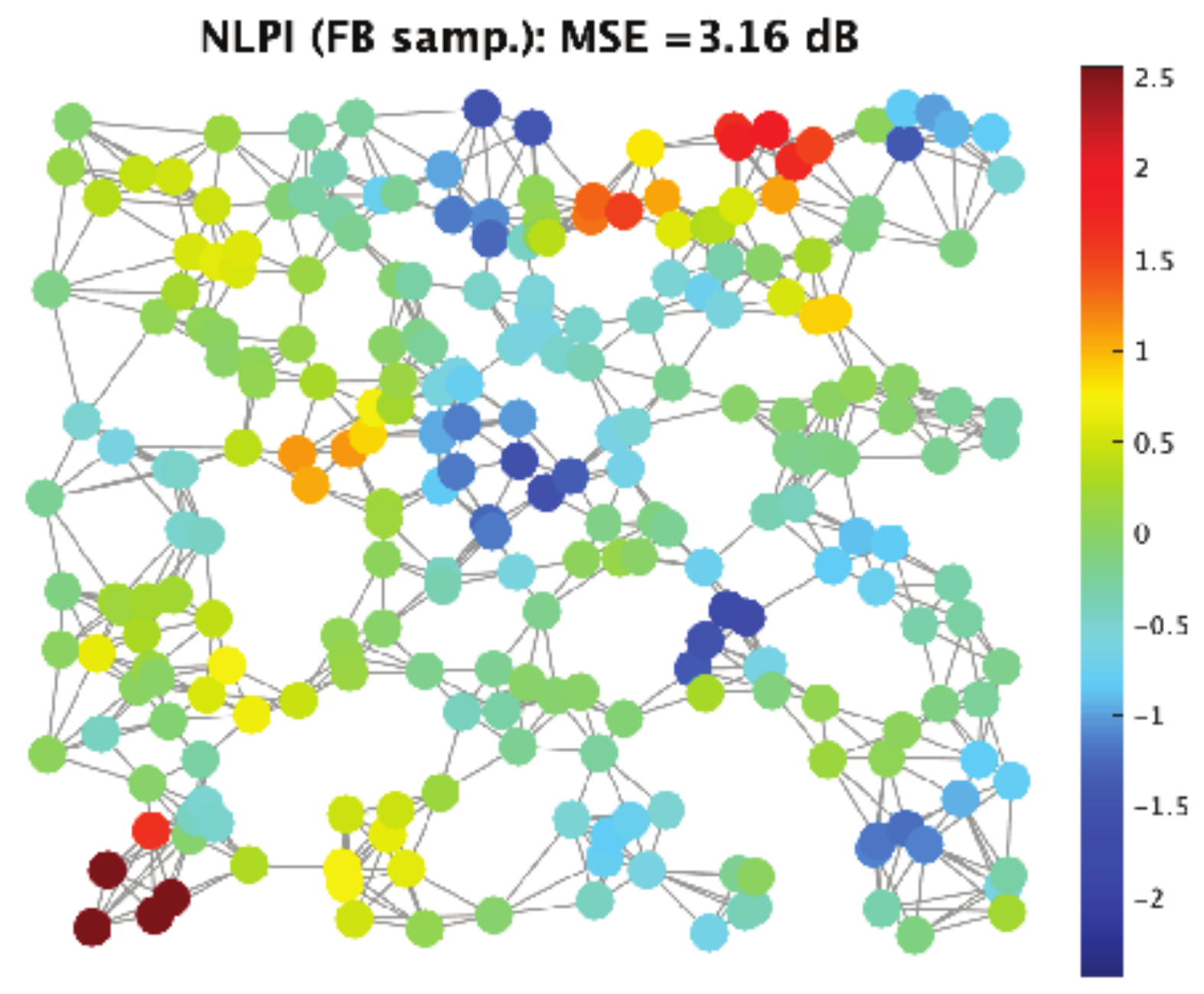} \label{ex1_con_sm_mx_bl}}
\caption{Signal recovery experiments for noisy graph signals on a random sensor graph with $N=256$. Sampling is performed by \eqref{eq:sampler_fb} in the vertex domain. UNC and PRE denote the unconstrained and predefined solutions. SM and ST denote recovery with smoothness and stochastic priors. LS and MX denote least-squares and minimax criteria.} \label{fig:ex1_recov}
\end{figure*}

\subsection{Synthetic Graph Signals}

In this subsection, we perform signal recovery experiments for synthetic graph signals.

\subsubsection{Sampling and Recovery Setting}
We perform signal recovery experiments of synthetic graph signals. The recovery framework is illustrated in Fig. \ref{gene_corrupt_filter_res}.

Here, we consider the following three graphs with $N=256$:
\begin{itemize}
\item Random sensor graph \footnote{Random sensor graphs are implemented by $k$ nearest neighbor graphs, whose vertices are randomly distributed in 2-D space $[0,1]\times [0,1]$ \cite{perraudin_gspbox:_2014}.};
\item Erd\H{o}s-R\'{e}nyi (ER) graph with edge connection probability $p=0.3$; and
\item 2D grid graph.
\end{itemize}
We use the following functions throughout the experiments.
\begin{itemize}
\item PSD function:
\begin{align}
\displaystyle
\widehat{\Gamma}_x(\lambda_i)=\exp\left(-\left\{\frac{2\lambda_i-\lambda_{\max}}{\sqrt{\lambda_{\max}}}\right\}^2\right).
\end{align}
\item Sampling function \#1 (for full-band sampling):
\begin{align}
S(\lambda_i)=
\begin{cases}
1&\lambda_{\max}/\lambda_{i}\leq 2\\
2-2\lambda_i/(\lambda_{\max})&\text{otherwise}.
\end{cases}\label{eq:sampler_fb}
\end{align}
\item Sampling function \#2 (for bandlimited sampling):
\begin{align}
S(\lambda_i)=
\begin{cases}
1& i\in[0,K-1],\\
0& \text{otherwise}.
\end{cases}\label{eq:sampler_bl}
\end{align}
\item Reconstruction function (for constrained recovery):
\begin{align}
\displaystyle
W(\lambda_i)=\cos\left(\frac{\pi}{2}\cdot\frac{\lambda_i}{\lambda_{\max}}\right).\label{eq:con_recon_kernel}
\end{align}
\item Smoothness function (for recovery with smoothness prior by \cite{tanaka_generalized_2020}):
\begin{align}
\displaystyle
V(\lambda_i)=\frac{\lambda_i}{\lambda_{\max}}+\varepsilon\label{eq:smoother}
\end{align}
where we set to $\varepsilon = 0.1$.
\end{itemize}
\label{exp}
Their spectral responses are shown in Fig. \ref{ex1_res_filter}.

A stochastic graph signal is generated by $\bm{x}\sim\mathcal{N}(\bm{0},\mb{\Gamma}_x)$, where $\mb{\Gamma}_x$ satisfies the GWSS conditions described in Section \ref{sec:samp_stationarity}, i.e., $\mb{\Gamma}_x=\mb{U}\widehat{\Gamma}_x(\mb{\Lambda})\mb{U}^*$. 

As mentioned in Appendix \ref{sec:stationarity_detail}, existing definitions of GWSS coincide when $\mb{\Gamma}_x$ is diagonalizable by $\mb{U}$.
Therefore, this setting simultaneously satisfies Definition \ref{definition:gwss2} and the other existing definitions of GWSS \cite{girault_stationary_2015,segarra_stationary_2017,perraudin_stationary_2017}. 

Under the presence of noise, we suppose that noise conforms to i.i.d. zero-mean Gaussian distribution with $\sigma^2=0.3$, i.e., $\eta \sim \mathcal{N}(0,0.3)$ and thus $\Gamma_\eta(\lambda_i)=0.3$ for all $i$. Graph signals are sampled with sampling ratio $N/K=4$.

We perform experiments for two cases. 1) The sampling function $S(\lambda_i)$ is full-band as in \eqref{eq:sampler_fb}. 2) $S(\lambda_i)$ is bandlimited as in \eqref{eq:sampler_bl}. In addition, two sampling domains, i.e., samplings in vertex and graph frequency domains, are considered. For vertex domain sampling, sampled vertices are selected randomly. 
For the predefined recovery, we use the given kernel $W(\lambda_i)$ in \eqref{eq:con_recon_kernel}.

We calculate the average MSE in 20 independent runs and compare with existing graph signal interpolation methods,
MKVV\cite{heimowitz_smooth_2018}, MKVD\cite{narang_signal_2013}, NLPD\cite{chen_discrete_2015}, GSOD\cite{segarra_reconstruction_2016} and NLPI\cite{narang_localized_2013}, and 
recovery methods under smoothness priors proposed in \cite{tanaka_generalized_2020}. Smoothness is measured by the energy of high frequency graph spectra with the measuring function $V(\lambda_i)$. We also show the result of the bandlimited reconstruction \cite{tanaka_spectral_2018}, i.e., reconstruction with the sampling filter \#2 without correction.

Since there has been no prior work on graph signal sampling with stochastic prior, we also compare with graph signal recovery with smoothness prior \cite{tanaka_generalized_2020} as a benchmark.

\subsubsection{Results}
Table \ref{tab:ex1} summarizes the average MSEs with the standard deviations (in decibels). Examples of recovered signals via full-band sampling are visualized in Fig. \ref{fig:ex1_recov}. 

From Table \ref{tab:ex1}, it is observed that the proposed methods consistently present the best MSEs in most cases.
It is regardless of the sampling domain. 
Sometimes, the MSEs for the bandlimited sampling is close to the proposed method, especially for the ER graph.
This is because of the eigenvalue distribution: The ER graph may produce eigenvalues sparsely distributed in low graph frequencies, i.e., $\lambda < \lambda_{\max}/2$ where we consider $\mb{L}$ as a graph operator,
and many eigenvalues exist in high graph frequencies ($\lambda \ge \lambda_{\max}/2$).
This implies that the cutoff frequency at the $K$th eigenvalue can be relatively high, resulting in wide bandwidth.
This also results in the fact that the bandlimited sampling filter in \eqref{eq:sampler_bl} may pass the spectra in the middle graph frequencies ($\lambda \in [\lambda_{\max}/4,3\lambda_{\max}/4]$);
This behavior improves the reconstruction MSEs.
Note that the proposed method presents stable recovery performances for all the graph and signal models considered.

The recovery with smoothness prior increases MSEs especially for vertex domain sampling because the PSD function $\widehat{\Gamma}_x(\lambda_i)$ does not represent smooth signals in this setting.
For the noiseless cases, the proposed methods demonstrate MSE improvements by 3--5 dB on average compared to the noisy case. 
The unconstrained solution of the proposed stochastic recovery results in the best performance for almost all cases.

\begin{table*}[t!]
\setlength{\tabcolsep}{0.3em}
\centering
\caption{Average MSEs of the reconstructed signals for the sea surface temperature data (in decibels). Notations are the same as Table \ref{tab:ex1}.}\label{tab:ex2}
\footnotesize
\begin{tabular}{c|c|cc|ccc|c|ccccc}
\hline
\multicolumn{2}{c|}{Prior} & \multicolumn{2}{c|}{Stochastic} & \multicolumn{3}{c|}{Smoothness \cite{tanaka_generalized_2020}} & \multirow{2}{*}{\begin{tabular}[c]{@{}c@{}}BL\\\cite{tanaka_spectral_2018}\end{tabular}} & \multirow{2}{*}{\begin{tabular}[c]{@{}c@{}}MKVV\\\cite{heimowitz_smooth_2018}\end{tabular}} & \multirow{2}{*}{\begin{tabular}[c]{@{}c@{}}MKVD\\\cite{narang_signal_2013}\end{tabular}} & \multirow{2}{*}{\begin{tabular}[c]{@{}c@{}}NLPD\\\cite{chen_discrete_2015}\end{tabular}} & \multirow{2}{*}{\begin{tabular}[c]{@{}c@{}}GSOD\\\cite{segarra_reconstruction_2016}\end{tabular}} & \multirow{2}{*}{\begin{tabular}[c]{@{}c@{}}NLPI\\\cite{narang_localized_2013}\end{tabular}}\\ \cline{1-7}
\multicolumn{2}{c|}{Reconst./ Criterion} & UNC & PRE & UNC & PRE LS & PRE MX&  &&&&&\\ \hline\hline
\multirow{2}{*}{BL} & VS & \textbf{20.6}$\pm$4.05 & 54.4$\pm$3.17 & 232$\pm$15.5 & 154$\pm$8.71 & 232$\pm$15.5 & 56.0$\pm$3.14 & 60.9$\pm$3.13 & 209$\pm$15.8 & 206$\pm$12.3 & 61.0$\pm$3.13 & 827$\pm$7.05\\
 & SS & \textbf{14.1}$\pm$1.16 & 55.1$\pm$0.15 & \textbf{14.1}$\pm$1.16 & 63.3$\pm$0.14 & 551$\pm$0.15 & 14.2$\pm$1.15 & - & - & - & - & -\\ \hline
 \multirow{2}{*}{FB} & VS & \textbf{28.3}$\pm$4.34 & 54.4$\pm$3.17 & 32.0$\pm$3.93 & 54.5$\pm$3.17 & 54.4$\pm$3.17 & -  & 60.9$\pm$3.13 & 210$\pm$15.8 & 206$\pm$12.0 & 61.0$\pm$3.13 & 828$\pm$6.88\\
 & SS & \textbf{19.5}$\pm$1.11 & 55.1$\pm$0.14 & 22.0$\pm$0.94 & 55.1$\pm$0.14 & 55.1$\pm$0.14 & - & - & - & - & - & -\\ \hline
\end{tabular}
\end{table*}

\begin{figure*}[t!]
\centering
 \subfigure[][Original]
  {\includegraphics[width=0.32\linewidth]{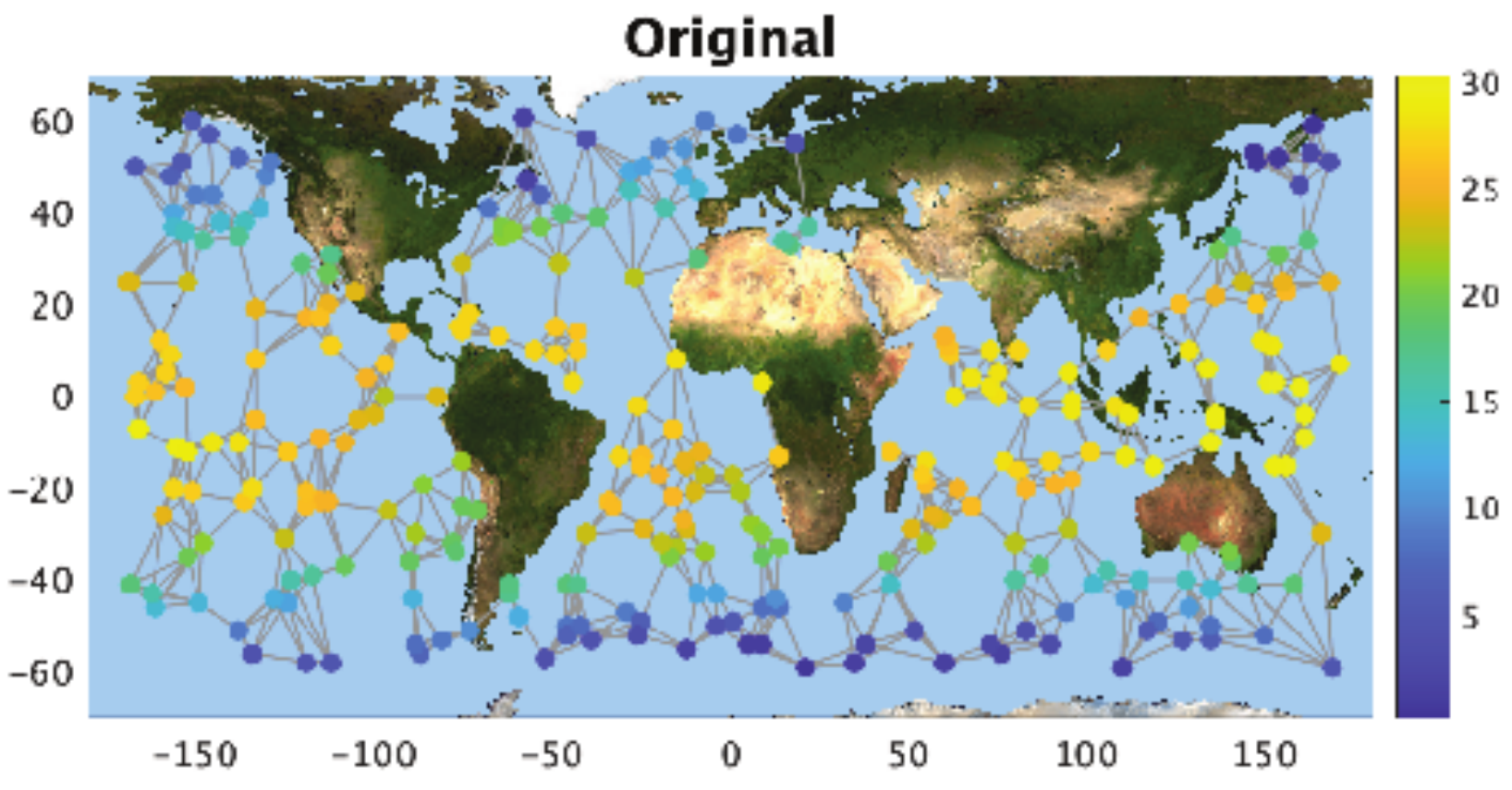} \label{ex1_ori}}
 \subfigure[][Proposed: UNC ST.]
  {\includegraphics[width=0.32\linewidth]{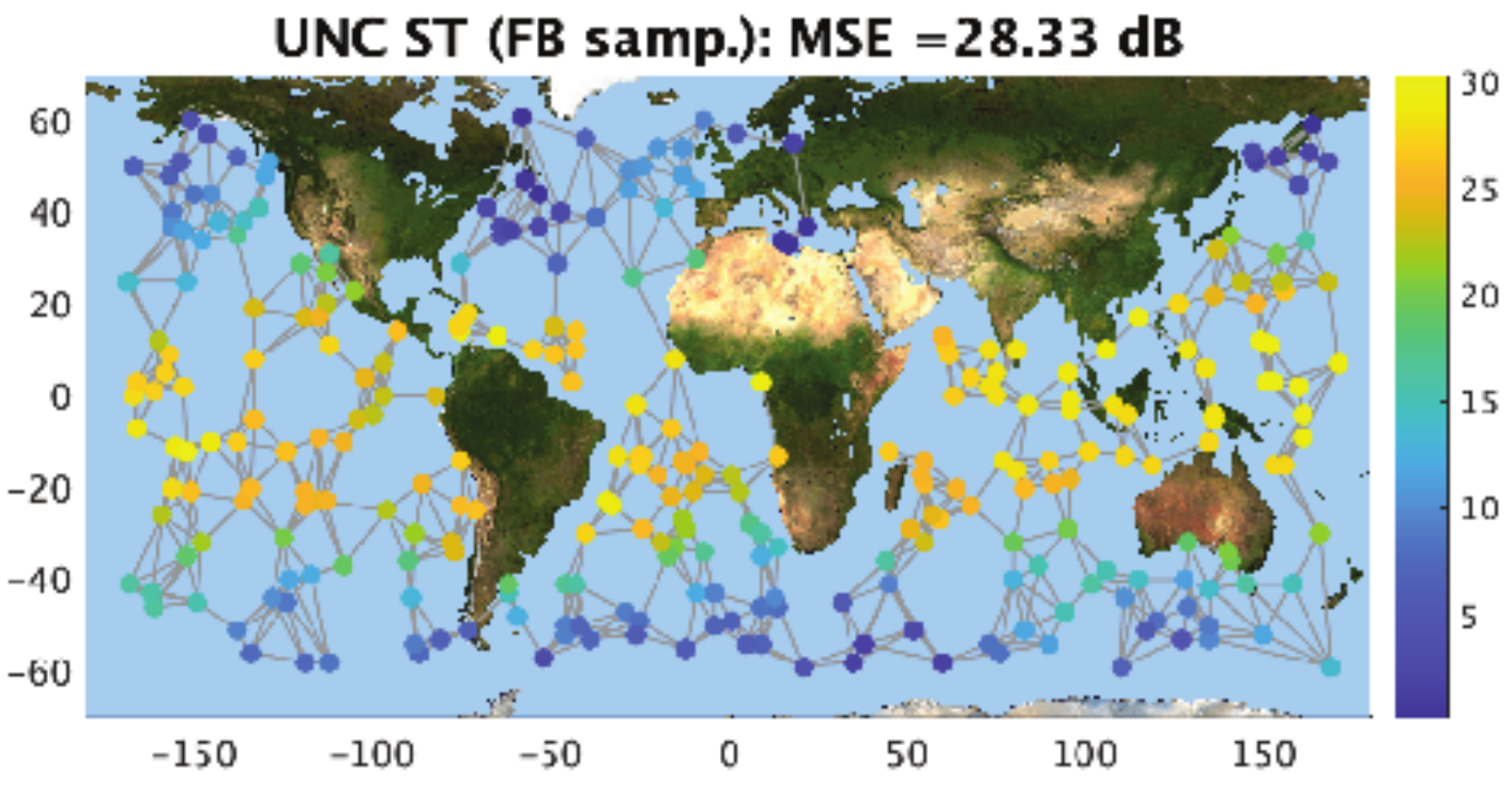} \label{ex1_unc_st_nbl}}
  \subfigure[][Proposed.: PRE ST.]
  {\includegraphics[width=0.32\linewidth]{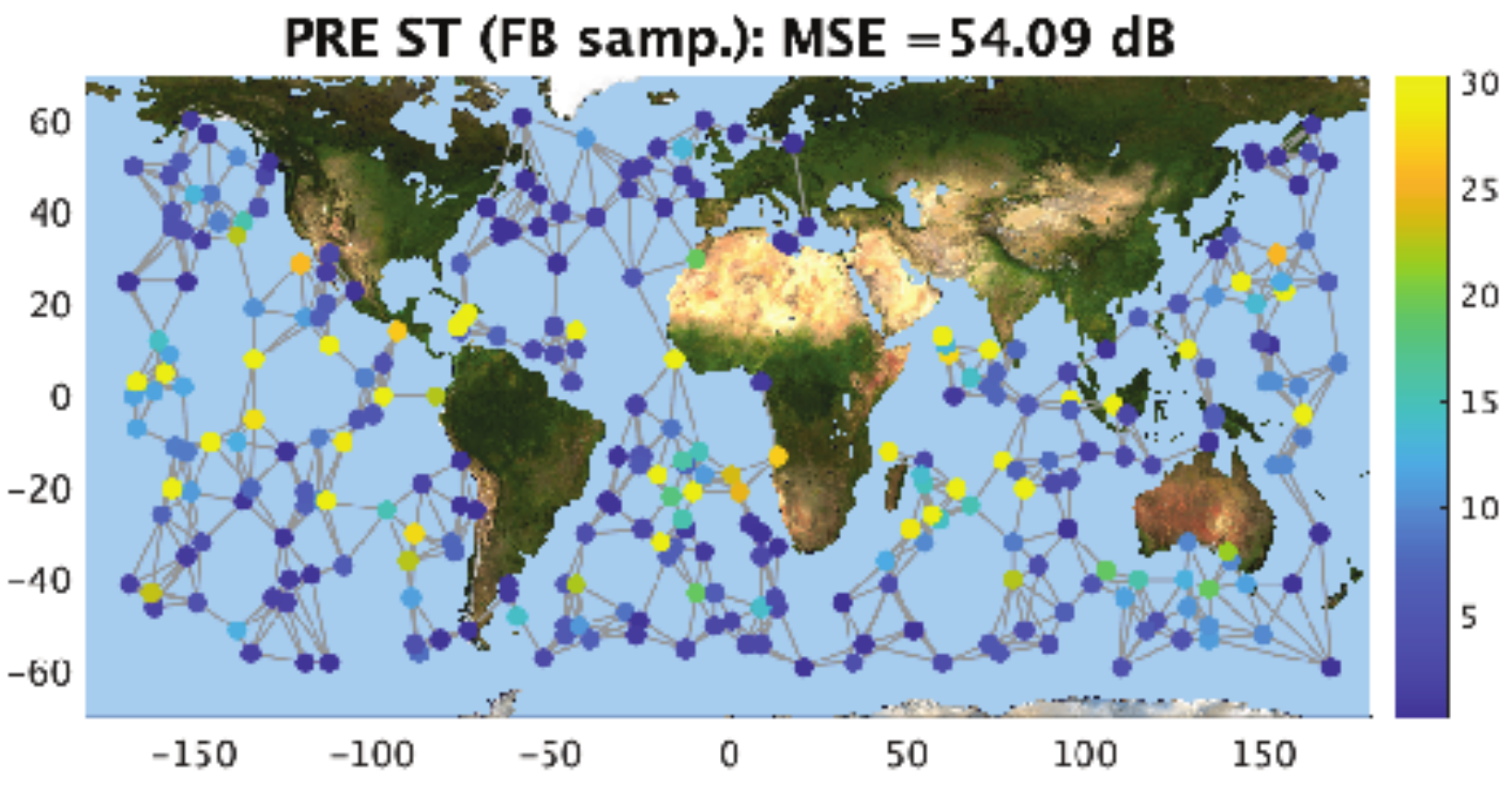} \label{ex1_con_st_nbl}}
 \subfigure[][UNC SM\cite{tanaka_generalized_2020}.]
  {\includegraphics[width=0.32\linewidth]{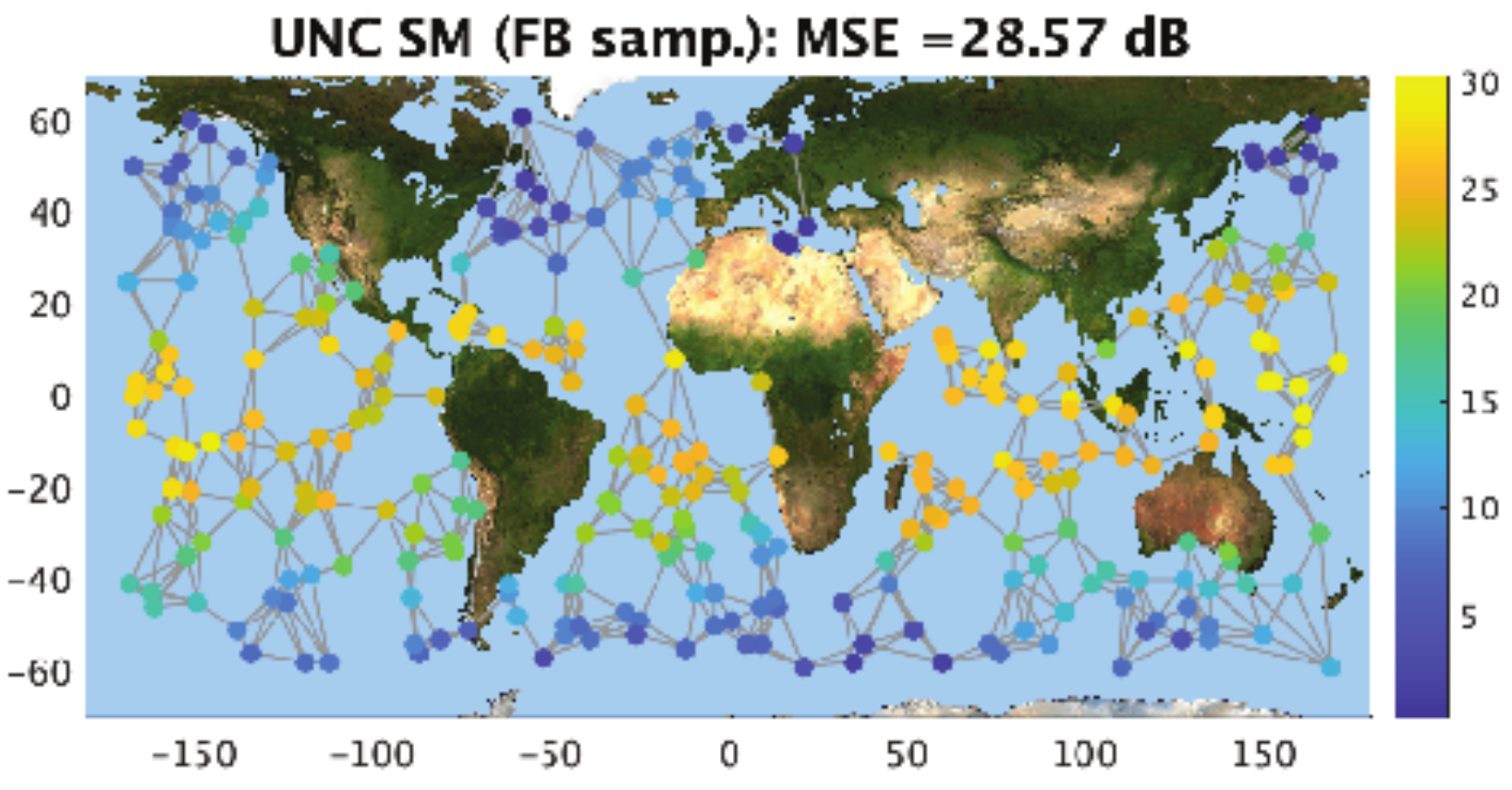} \label{ex1_unc_sm_nbl}}
 \subfigure[][PRE SM LS\cite{tanaka_generalized_2020}.]
  {\includegraphics[width=0.32\linewidth]{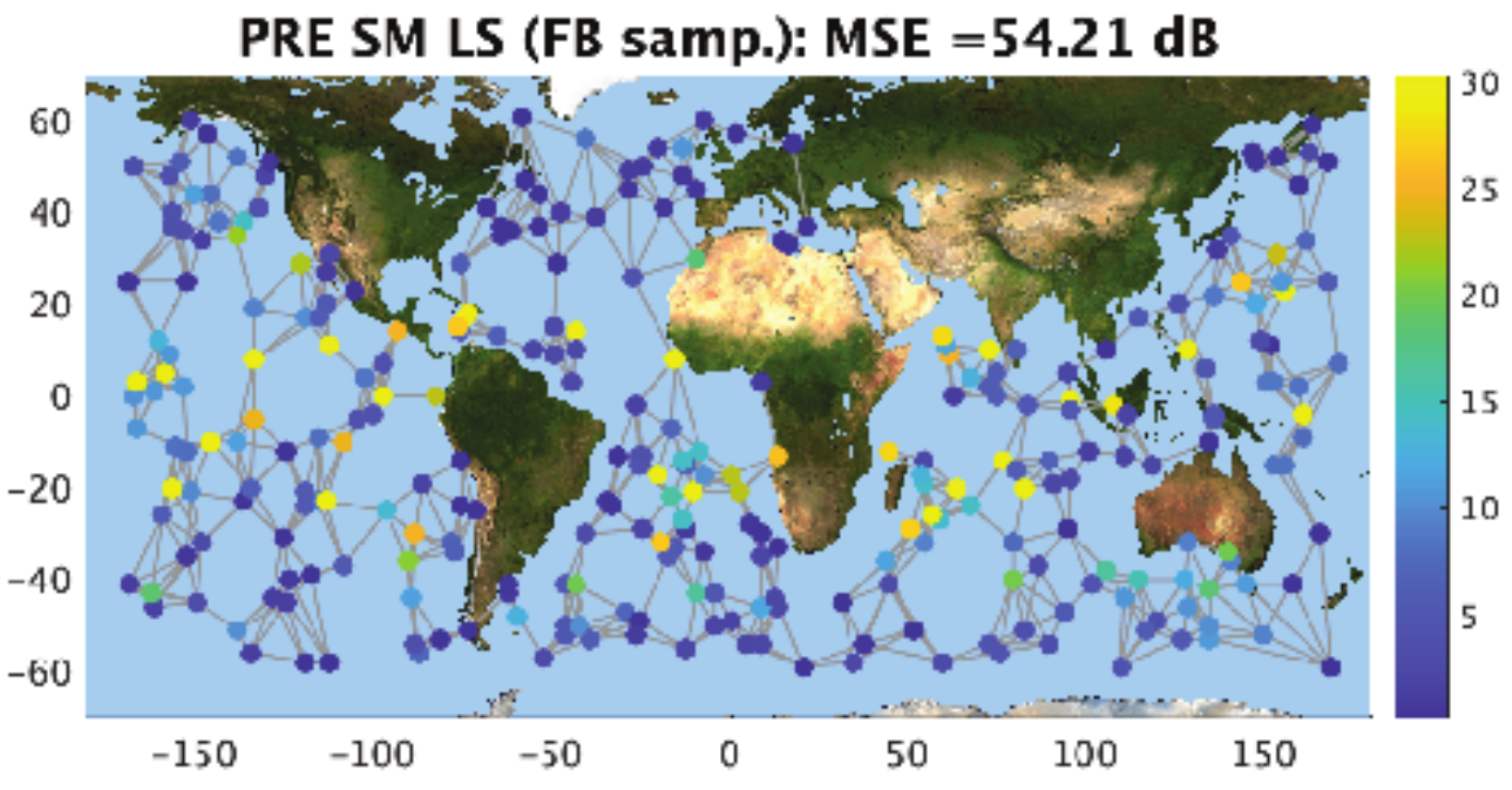} \label{ex1_con_sm_ls_nbl}}
 \subfigure[][PRE SM MX\cite{tanaka_generalized_2020}.]
  {\includegraphics[width=0.32\linewidth]{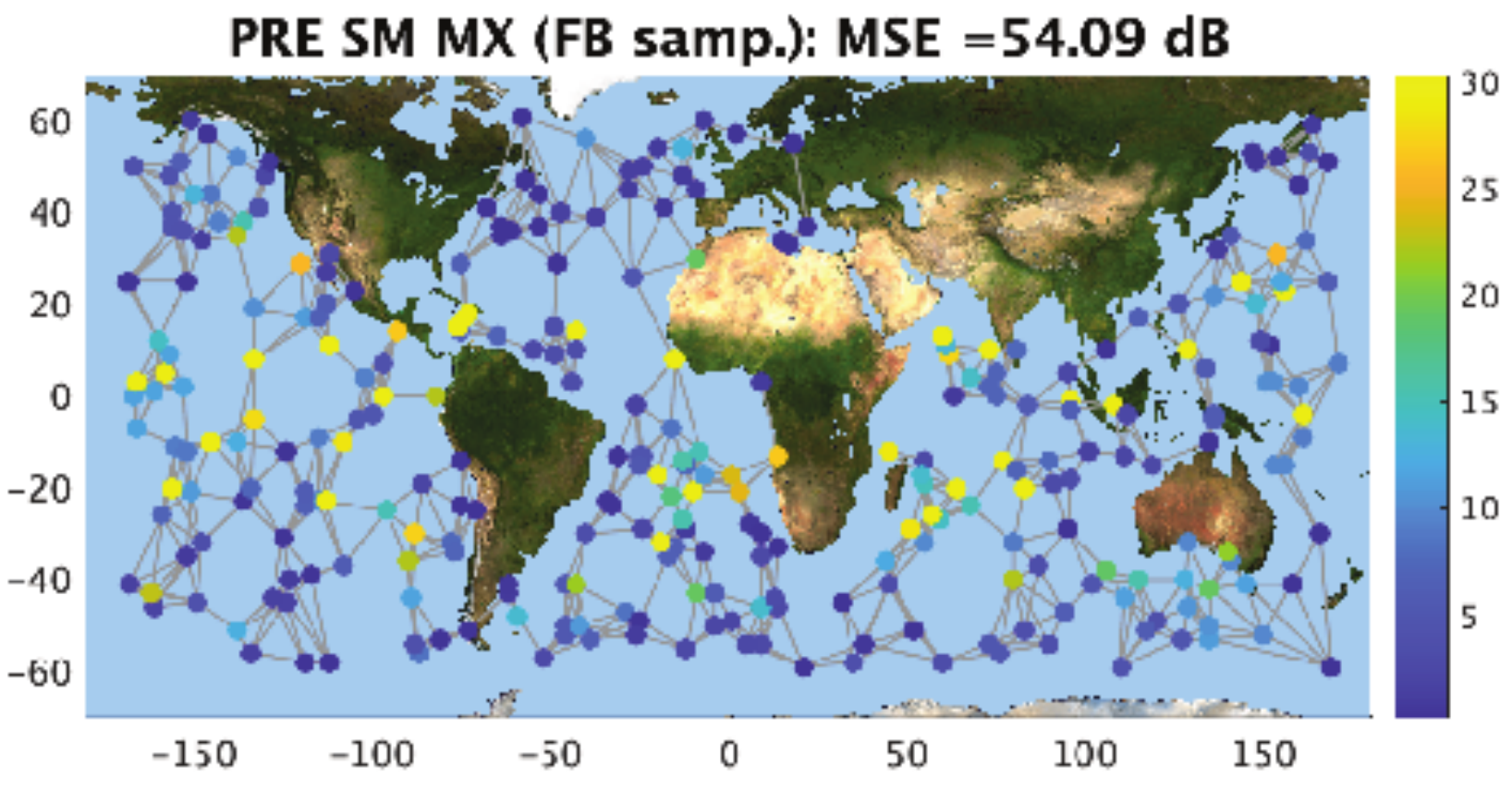} \label{ex1_con_sm_mx_nbl}}
 \subfigure[][BL\cite{tanaka_spectral_2018}.]
  {\includegraphics[width=0.32\linewidth]{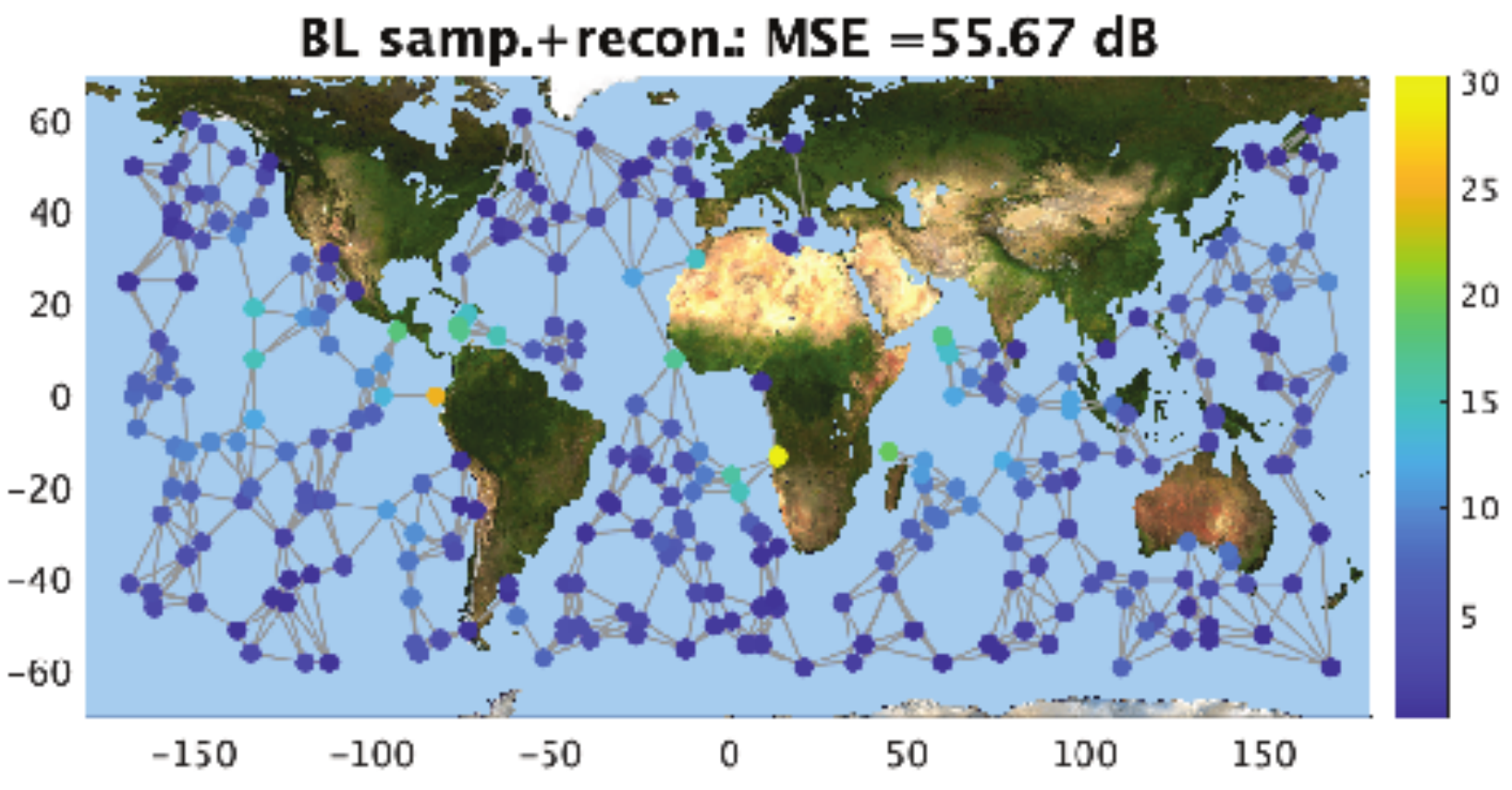} \label{ex1_bl}}
 \subfigure[][MKVV\cite{heimowitz_smooth_2018}.]
  {\includegraphics[width=0.32\linewidth]{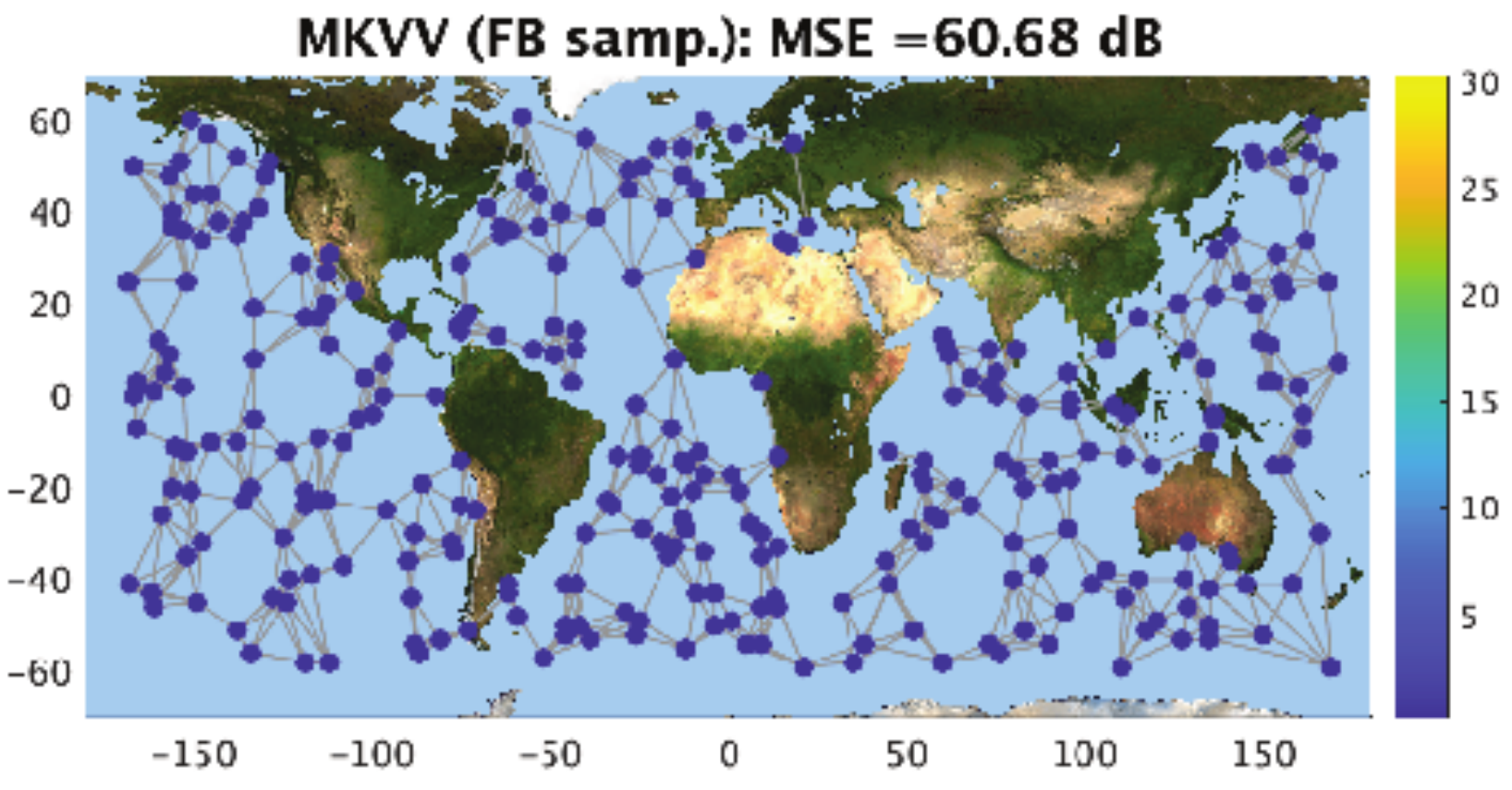} \label{ex1_unc_st_bl}}
 \subfigure[][MKVD\cite{narang_signal_2013}.]
  {\includegraphics[width=0.32\linewidth]{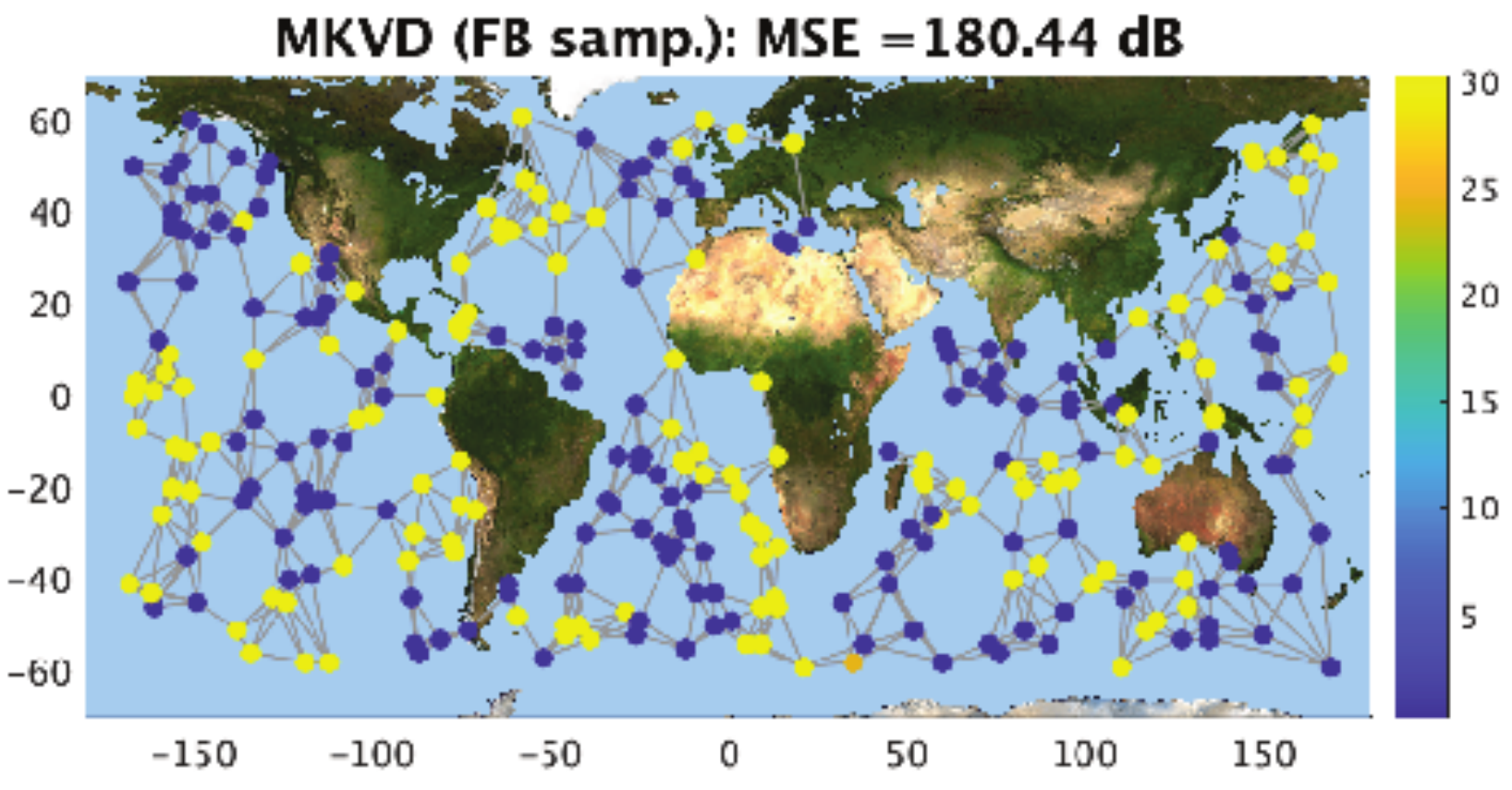} \label{ex1_con_st_bl}}
 \subfigure[][NLPD\cite{chen_discrete_2015}.]
  {\includegraphics[width=0.32\linewidth]{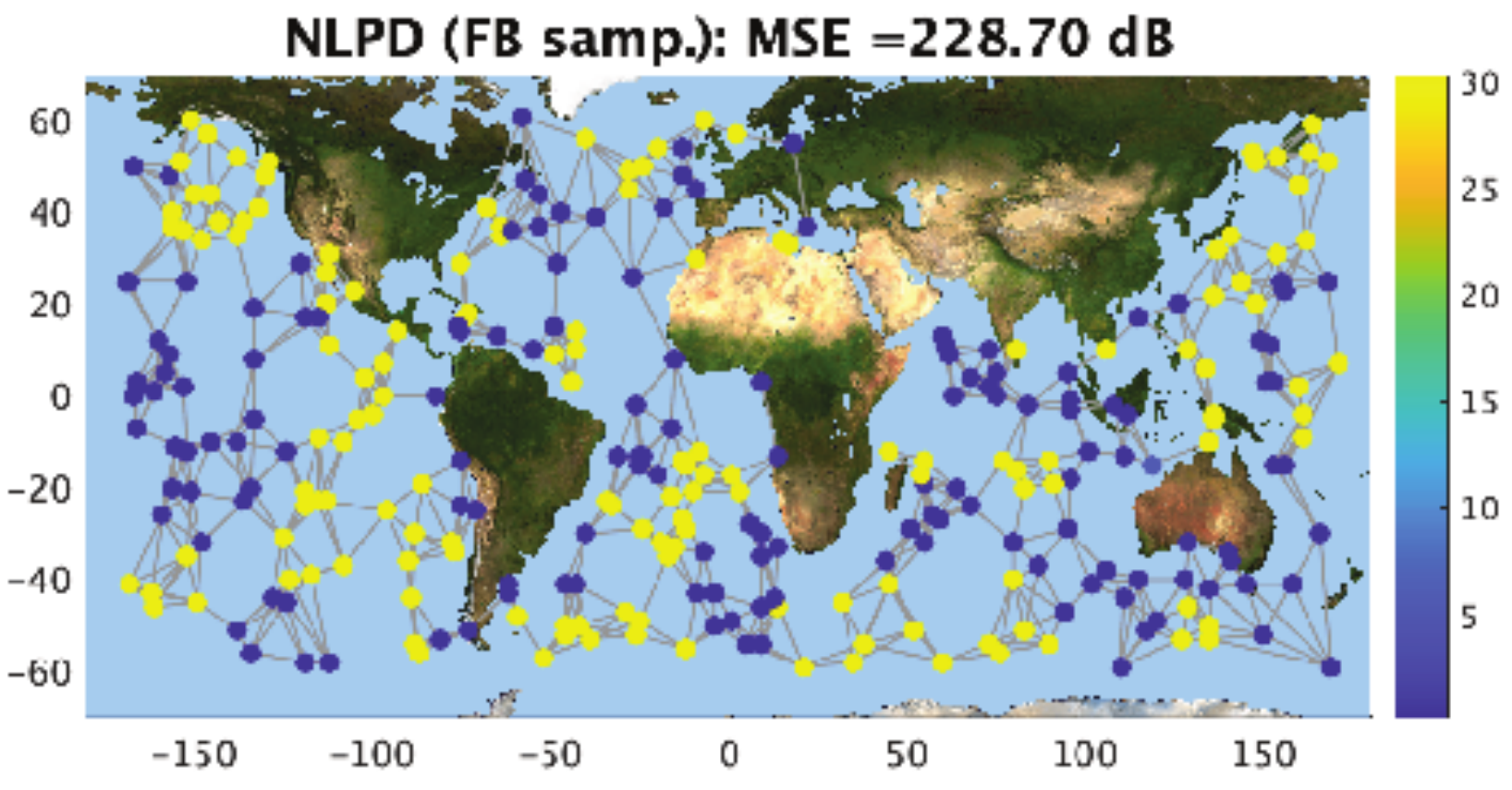} \label{ex1_unc_sm_bl}}
 \subfigure[][GSOD\cite{segarra_reconstruction_2016}.]
  {\includegraphics[width=0.32\linewidth]{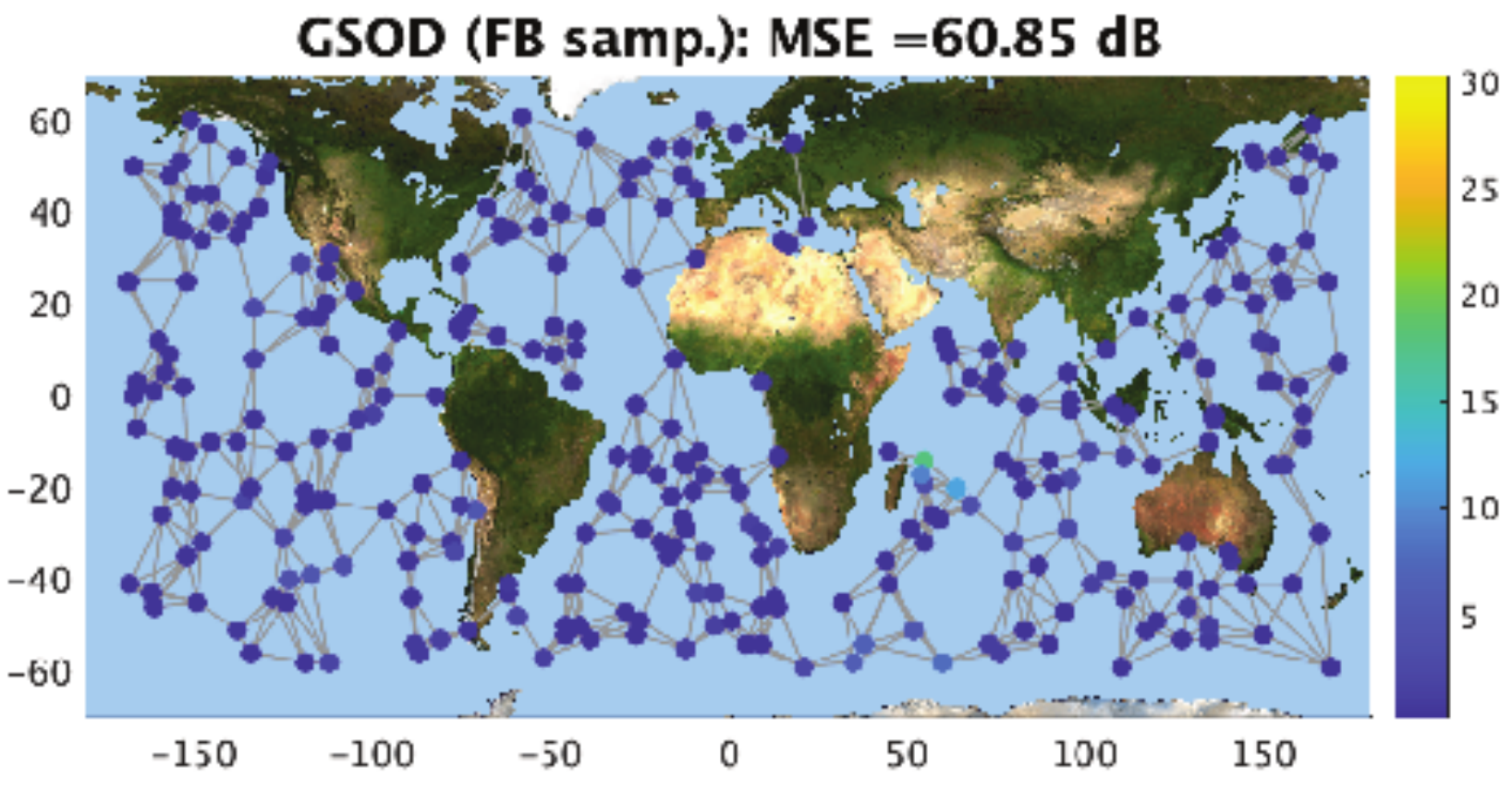} \label{ex1_con_sm_ls_bl}}
 \subfigure[][NLPI\cite{narang_localized_2013}. ]
  {\includegraphics[width=0.32\linewidth]{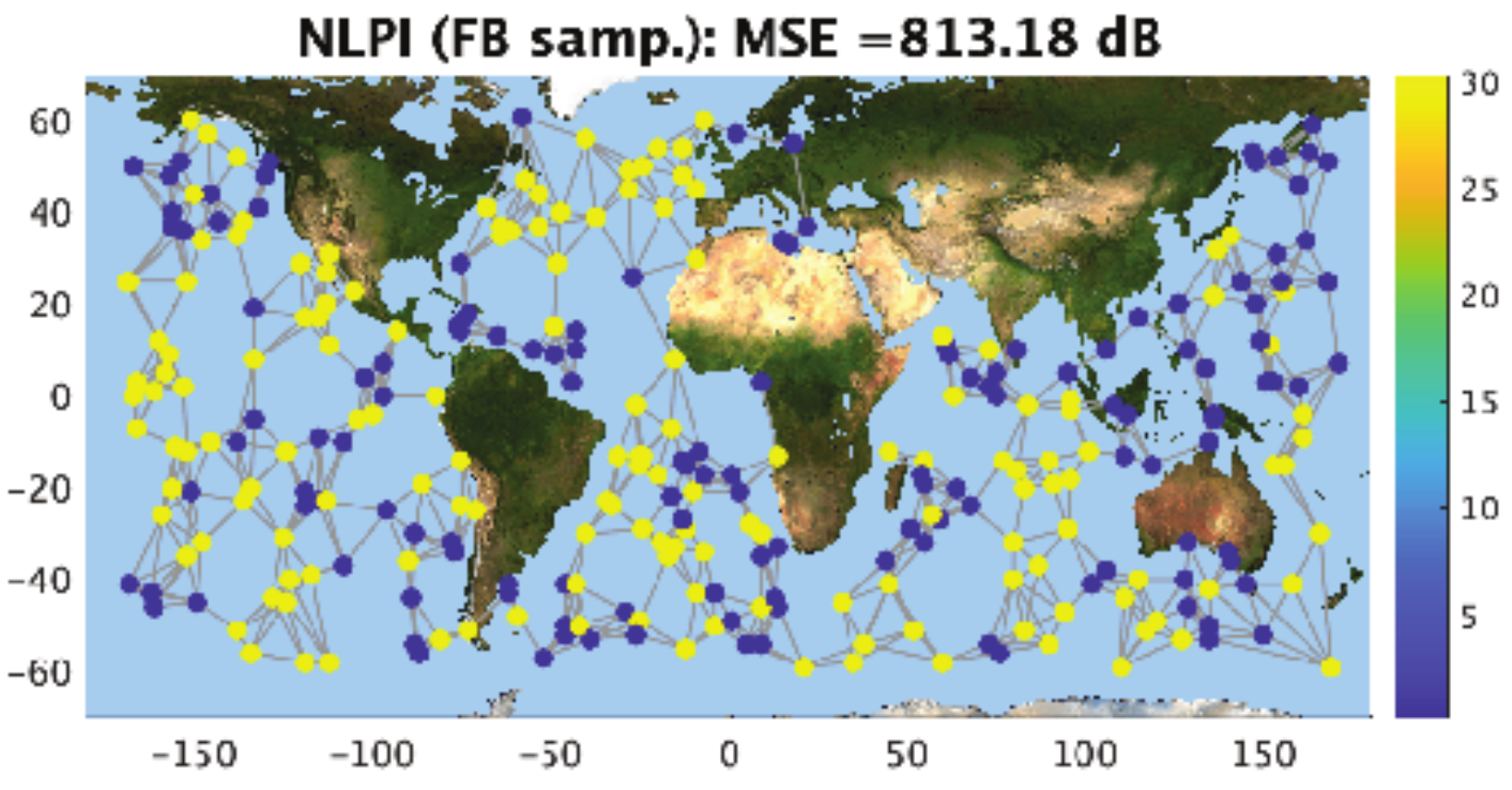} \label{ex1_con_sm_mx_bl}}
\caption{Signal recovery experiments for the sea surface temperature data on a $5$-nearest neighbor graph. Sampling is performed by \eqref{eq:sampler_fb} in the vertex domain. Abbreviations are the same as those in Fig. \ref{fig:ex1_recov}.} \label{fig:ex2_recov_sensor}
\end{figure*}

\subsection{Real-world  Data}\label{subsec:real_experiments}

We also perform signal recovery experiments on real-world data.

\subsubsection{Sampling and Recovery Setting}

We use the global sea surface temperature dataset \cite{rayner_global_2003}. It records snapshots of sea surface temperatures for every month from 2004 to 2021. They are spatially sampled at intersections of 1-degree latitude-longitude grids.
For the experiment, we used data for 24 months from 2018 to 2020.
For this dataset, we randomly sample $N=291$ intersections and they are regarded as vertices.
We then construct a $5$-nearest neighbor graph based on Euclidean distance.
Edge weights are set to be unweighted.
We then remove the edges crossing land areas.
The number of samples is set to $K=73$.

For each time instance, we estimate the covariance from the data in the previous year.
We use the method proposed in \cite{perraudin_stationary_2017}.
The estimated covariance is used for the proposed graph Wiener filter.
We calculate the average MSEs of 24 months for samplings in vertex and graph frequency domains, and compare with the existing methods in the previous subsection.
For vertex domain sampling, we perform the experiments on 20 random sampling patterns of vertices for calculating the average MSE.

\subsubsection{Results}

Table \ref{tab:ex2} summarizes the average MSEs with the standard deviations (in decibels). Recovered graph signals are visualized in Fig. \ref{fig:ex2_recov_sensor}.

From Table \ref{tab:ex2}, we can observe that the proposed method outperforms existing methods for almost all cases, as expected to the experiment on synthetic data.
This is also validated in Fig. \ref{fig:ex2_recov_sensor}.

\section{Conclusion}\label{sec:conclusion}

Generalized graph signal sampling with stochastic priors is proposed in this paper. We derive the graph Wiener filter for unconstrained and predefined reconstruction, based on minimization of the expectation of the squared norm error between the original and reconstructed graph signals. We show that, when sampling is performed in the graph frequency domain, spectral responses of the graph Wiener filter parallel those in generalized sampling for WSS signals.
We also reveal a theoretical relationship between the proposed graph Wiener filter and existing signal recovery under different priors.
In the recovery experiments, we validate the MSE improvement of the proposed methods for various graphs and two sampling domains.

\begin{appendices}

\section{Wide Sense Stationarity}\label{sec:stationarity}

Stationarity of time-domain signals is reviewed here since this is useful to understand the connection between WSS and GWSS.

We consider a continuous-time signal $x(t)$. Its stationarity is defined as follows:

\begin{definition}[Wide sense stationary for time-domain signals]\label{definition:wss}
Let $x(t)$ and $\gamma_x(t)$ be a stochastic signal in the time domain and its autocovariance function, respectively. The signal $x(t)$ is a wide sense stationary process if and only if the following two conditions are satisfied: 
\begin{enumerate}
\itemequation[cond1:wss]{}{\mathbb{E}\left[x(t)\right]=\mu_x=\mathrm{const},}
\itemequation[cond2:wss]{}{\displaystyle \mathbb{E}\left[(x(t)-\mu_x)(x(\tau)-\mu_x)^*\right]=\gamma_x(t-\tau).}
\end{enumerate}
\end{definition}

\noindent A WSS process can be characterized in the Fourier domain by Wiener-Khinchin theorem \cite{wiener_generalized_1930}. If $x(t)$ is a WSS process, then its power spectral density (PSD) function coincides with the CTFT of the autocovariance function of $x(t)$, $\gamma_x(t)\in L_1$, i.e., 
\begin{align}
\Gamma_x(\omega)=\int_{-\infty}^{\infty} \gamma_x(t)e^{-j\omega t}d\omega.
\end{align}

Not surprisingly, WSS has several equivalent expressions because of the correspondence between time-shift and frequency-modulation, i.e., $x(t-\tau)\leftrightarrow e^{-j\tau\omega}X(\omega)$. 

Here, we present another two expressions of WSS to show the connection with GWSSs.
\begin{corollary}[WSS by Shift]\label{cor:wss_shift}
Let $T_{t_0}\{\cdot\}$ be the shift operator that delays the signal by $t_0$, i.e., $T_{t_0}\{x(t)\}=x(t-t_0)$. Definition \ref{definition:wss} can be written as follows:
\begin{enumerate}
\itemequation[cond1:wss1]{}{\mathbb{E}\left[T_{t_0}\{x(t)\}\right]=\mu_x=\mathrm{const},}
\item $\displaystyle\mathbb{E}\left[(T_{t_0} \{x(t)\}-\mu_x)(T_{t_0} \{x(\tau)\}-\mu_x)^*\right]$
\nonumitemequation[cond2:wss1]{}{
=\gamma_x(t-\tau).}
\end{enumerate}

\end{corollary}
\noindent The equivalence of conditions \eqref{cond1:wss1} and \eqref{cond1:wss} is easy to verify. Since the autocovariance only depends on the time difference, \eqref{cond2:wss1} is identical to \eqref{cond2:wss}.

\begin{table*}[t]
\centering
\footnotesize
\caption{Comparison among WSS and GWSSs. OPE and COV denote operator and covariance, respectively.}\label{table:wss_gwss}
\begin{tabular}{c|A|B||c|A|B}
\hline
\multirow{6}{*}{\textbf{OPE}} & \multicolumn{2}{c||}{WSS} & \multirow6{*}{\textbf{Mean}} & \multicolumn{2}{c}{WSS} \\ \hhline{~|--||~|--}
 & \multicolumn{1}{c|}{Shift} & \multicolumn{1}{c||}{Modulation} &  & \multicolumn{1}{c|}{Shift} & \multicolumn{1}{c}{Modulation} \\ \hhline{~|--||~|--}
 & $T_{t_0}\{x(t)\}=x(t-t_0)$ & $x(t_0)*\delta_t=x(t)$ &  & $\mathbb{E}\left[T_{t_0}\{x(t)\}\right]=\mu_x$ & $\mathbb{E}\left[x(t_0)*\delta_\tau\right]=\mu_x$  \\ \hhline{~|--||~|--}
 & \multicolumn{2}{c||}{GWSS} &  & \multicolumn{2}{c}{GWSS} \\ \hhline{~|--||~|--}
 & \multicolumn{1}{c|}{GWSS$_\text{T}$} & \multicolumn{1}{c||}{GWSS$_\text{M}$} &  & \multicolumn{1}{c|}{GWSS$_\text{T}$} & \multicolumn{1}{c}{GWSS$_\text{M}$} \\ \hhline{~|--||~|--} 
 &  $\mb{T}_{\mathcal{G}}\bm{x}=\mb{U}\exp(j\mb{\Pi})\mb{U}^*\bm{x}$ & $\bm{x}\star\bm{\delta}_n=\mb{M}\text{diag}(\mb{U}^*\bm{\delta}_n)\mb{U}^*\bm{x}$  &  &$\mathbb{E}[(\mb{T}_{\mathcal{G}}\bm{x})_n]=\mu_x$  & $\mathbb{E}[(\bm{x}\star \bm{\delta}_n)_m]=\mu_x$  \\ \hhline{===::===}
\multirow{8}{*}{\textbf{COV}} & \multicolumn{2}{c||}{WSS} & \multirow{8}{*}{\textbf{PSD}} & \multicolumn{2}{c}{WSS} \\ \hhline{~|--||~|--} 
 & \multicolumn{1}{c|}{Shift} & \multicolumn{1}{c||}{Modulation} &  & \multicolumn{1}{c|}{Shift} & \multicolumn{1}{c}{Modulation} \\ \hhline{~|--||~|--}  
 &  $\displaystyle\mathbb{E}\left[T_{t_0} \{\bar{x}(t)\}T_{t_0} \{\bar{x}(\tau)\}^*\right]$ & $\displaystyle\mathbb{E}\left[(\bar{x}(t_0)*\delta_t)(\bar{x}(t_0)*\delta_\tau)^*\right]$  &  & \multicolumn{1}{D}{$\Gamma_x(\omega)=$}&\multicolumn{1}{E}{\hspace{-1em}$\int_{-\infty}^{\infty}\gamma_x(t)e^{-j\omega t}d\omega$}  \\  
 & $=\gamma_x(t-\tau)$ & $=\gamma_x(t-\tau)$ & & \multicolumn{1}{D}{}&\multicolumn{1}{E}{\hspace{-2em}always}  \\ \hhline{~|--||~|--}
 & \multicolumn{2}{c||}{GWSS} &  & \multicolumn{2}{c}{GWSS} \\ \hhline{~|--||~|--}
 & \multicolumn{1}{c|}{GWSS$_\text{T}$} & \multicolumn{1}{c||}{GWSS$_\text{M}$} &  & \multicolumn{1}{c|}{GWSS$_\text{T}$} & \multicolumn{1}{c}{GWSS$_\text{M}$} \\ \hhline{~|--||~|--}  
 & $\mathbb{E}[(\mb{T}_{\mathcal{G}}\bar{\bm{x}})_n(\mb{T}_{\mathcal{G}}\bar{\bm{x}})^*_k]$ & $\mathbb{E}[(\bar{\bm{x}}\star\bm{\delta}_n)_m(\bar{\bm{x}}\star\bm{\delta}_k)^*_m]$ &  & $\widehat{\Gamma}_x(\mb{\Lambda})=\mb{U}^*\mb{\Gamma}_x\mb{U}$ & $\widehat{\Gamma}_x(\mb{\Lambda})=\mb{U}^*\mb{\Gamma}_x\mb{U}$  \\
 & $=[\mb{\Gamma}_x]_{n,k}$ & $=[\mb{\Gamma}_x]_{n,k}$ & & if $\mb{\Lambda}$ is distinct & always\\ \hline
\end{tabular}
\\ We denote by $\delta_t=\delta(t_0-t)$, $\bar{x}(t)=x(t)-\mu_x$, $\bar{\bm{x}}=\bm{x}-\eta_x\bm{1}$, and $\mb{\Pi}=\pi\sqrt{\mb{\Lambda}/\rho_{\mathcal{G}}}$.
\end{table*}

Noting that
\begin{align}
x(\tau)
&=x(t)*\delta(t-\tau),\label{eq:shiftbydelta}
\end{align}
leads to the following corollary.
\begin{corollary}[WSS by Modulation]\label{cor:wss_localize}
Definition \ref{definition:wss} can be expressed equivalently by using $\delta(t)$ as follows:
\begin{enumerate}
\itemequation[cond1:wss2]{}{\mathbb{E}\left[x(t_0)*\delta(t_0-\tau)\right]=\mu_x=\mathrm{const},}
\item $\displaystyle\mathbb{E}\left[(x(t_0)*\delta(t_0-t)-\mu_x)(x(t_0)*\delta(t_0-\tau)-\mu_x)^*\right]$
\nonumitemequation[cond2:wss2]{}{=\gamma_x(t-\tau).}
\end{enumerate}

\end{corollary}

\noindent In the paper, we utilize these expressions of WSS for formally defining GWSS.

\section{Graph Wide Sense Stationarity}
\label{sec:stationarity_detail}
In this Appendix, we describe and compare some definitions of GWSS, including ours.
These definitions differ in whether Corollary \ref{cor:wss_shift} or \ref{cor:wss_localize} in Appendix \ref{sec:stationarity} is used for the baseline.
They coincide for time-domain signals, but this is not the case for graph signals,
leading to slightly different definitions of GWSS. 
The definitions are mainly divided according to whether the covariance is diagonalizable by the GFT basis $\mb{U}$. These differences are summarized in Table \ref{table:wss_gwss}.

\subsection{GWSS followed by Corollary \ref{cor:wss_shift}}\label{app:gwss_type1}
First, we show the definition of GWSS\cite{girault_stationary_2015} as a counterpart of Corollary \ref{cor:wss_shift}. 
In the literature of GSP, the graph Laplacian $\mb{L}$ is often referred to as a counterpart of the time-shift operator (translation operator) \cite{segarra_stationary_2017,marques_sampling_2016}. However, $\mb{L}$ changes the signal energy, i.e., $\|\mb{L}\bm{x}\|\neq\|\bm{x}\|$, while time shift does not, i.e., $\|T_\tau\{x[n]\}\|=\|x[n]\|$. 

The first definition of GWSS introduced here is based on a graph-translation operator which preserves the signal energy.

\begin{definition}[Graph Wide Sense Stationary by Translation (GWSS$_\text{T}$) \cite{girault_signal_2015}]\label{definition:gwss1}
Let $\bm{x}$ be a graph signal on a graph $\mathcal{G}$. 
Suppose that a graph-translation operator $\mb{T}_{\mathcal{G}}$ is defined by
\begin{align}
\mb{T}_{\mathcal{G}}\coloneqq \mathrm{exp}\left(j\pi\sqrt{\frac{\mb{L}}{\rho_\mathcal{G}}}\right),\label{eq:giraultgso}
\end{align}
where $\rho_\mathcal{G}=\max_{m\in\mathcal{V}}\sqrt{2d_m(d_m+\overline{d}_m)}$, $\overline{d}_m=\frac{\sum_{n=0}^{N-1}a_{mn}d_n}{d_m}$ and $\rho_\mathcal{G}\geq \lambda_{\max}$ \cite{das_extremal_2011}. 
Then, $\bm{x}$ is a graph wide sense stationary process by translation, if and only if it the following two conditions are satisfied,
\begin{enumerate}
\itemequation[cond1:gwss1]{}{\mathbb{E}\left[(\mb{T}_{\mathcal{G}}\bm{x})_n\right]=\mu_x=\mathrm{const},}
\itemequation[cond2:gwss1]{}{\mathbb{E}\left[(\mb{T}_{\mathcal{G}}\bm{x}-\mu_x\bm{1})_n(\mb{T}_{\mathcal{G}}\bm{x}-\mu_x\bm{1})_k^*\right]=[\mathbf{\Gamma}_x]_{n,k}.}
\end{enumerate}
\end{definition}

\noindent The condition \eqref{cond1:gwss1} implies that $\mathbb{E}[\bm{x}]$ is also constant. To see this, note that \eqref{cond1:gwss1} can be expressed as 
\begin{align}
\mb{T}_\mathcal{G}\mathbb{E}[\bm{x}]&=\mb{U}\exp(j\pi\sqrt{\mb{\Lambda}/\rho_\mathcal{G}})\mb{U}^*\mathbb{E}[\bm{x}].\label{cond1:gwss1_mean}
\end{align}
Since $\lambda_0=0$ and $\bm{u}_0=\bm{1}$ are always satisfied, i.e., $\mb{L}\bm{1}=0\cdot \bm{1}$, \eqref{cond1:gwss1_mean} holds if and only if 
$\mathbb{E}[\bm{x}]= \mu_x \bm{u}_{0}=\mu_x\bm{1}$. It also implies that $\mb{T}_\mathcal{G}\bm{1}=\bm{1}$.

When $\{\lambda_i\}_{i=0,\ldots,N-1}$ are distinct, the condition \eqref{cond2:gwss1} implies that $\mb{\Gamma}_x$ is diagonalizable by $\mb{U}$ because \eqref{cond2:gwss1} can be expressed as 
\begin{align}
\mb{\Gamma}_x&=\mb{T}_\mathcal{G}\mathbb{E}[(\bm{x}-\mu_x\bm{1})(\bm{x}-\mu_x\bm{1})^*]\mb{T}_\mathcal{G}^*&\nonumber\\
&=\mb{U}\exp(j\pi\sqrt{\mb{\Lambda}/\rho_\mathcal{G}})\mb{U}^*\mb{\Gamma}_x\mb{U}\exp(-j\pi\sqrt{\mb{\Lambda}/\rho_\mathcal{G}})\mb{U}^*\nonumber\\
&=\mb{U}(\widehat{\mb{\Gamma}}_x\circ \mb{\Theta})\mb{U}^*,
\label{proof:rescov2}
\end{align}
where $\widehat{\mb{\Gamma}}_x= \mb{U}^*\mb{\Gamma}_x\mb{U}$ and $[\mb{\Theta}]_{i,l}=\exp\{j\pi (\sqrt{\lambda_i/\rho_{\mathcal{G}}}-\sqrt{\lambda_l/\rho_{\mathcal{G}}})\}$. The third equality follows by the relationship $\mathrm{diag}(\bm{a})\mb{X}\mathrm{diag}(\bm{b}^*)=\mb{X}\circ \bm{a}\bm{b}^*$.
Since $[\mb{\Theta}]_{i,l}=1$ for $\lambda_i=\lambda_l$ and $[\mb{\Theta}]_{i,l}\neq 1$ otherwise, the equality holds if and only if $[\widehat{\mb{\Gamma}}_x]_{i,l}=0$ for $\lambda_i\neq\lambda_l$. If the eigenvalues are distinct, then this condition is equivalent that $\mb{\Gamma}_x$ is diagonalizable by $\mb{U}$.

\subsection{GWSS followed by Corollary \ref{cor:wss_localize}}\label{app:gwss_type2}

Next, we study the properties of Definition \ref{definition:gwss2} for GWSS used in our generalized sampling, as a counterpart of Corollary \ref{cor:wss_localize}.
In Definition \ref{definition:gwss2},
the condition in \eqref{cond1:gwss2} is identical to
\begin{align}
\mathbb{E}[\bm{x}\star \bm{\delta}_n]&=\mb{M}\mathrm{diag}(u_{0}[n],u_{1}[n],\ldots)\mb{U}^*\mathbb{E}[\bm{x}].\label{cond1:gwss2_spectrum}
\end{align}
Since $\mb{M}\mathrm{diag}(1,0,\ldots,0)=[\bm{1}\, \bm{0}\,\cdots\,\bm{0}]$, the equality holds if and only if $\mathbb{E}[\bm{x}]=\mu_x \bm{u}_{0}=\mu_x\bm{1}$. Therefore, $\cdot\star\bm{\delta}_n$ does not change the mean of graph signals. 
The condition in \eqref{cond2:gwss2} implies the covariance $\mb{\Gamma}_x$ has to be diagonalizable by $\mb{U}$. We show this fact in the following lemma.

\begin{lemma}\label{lemma:gwssm}
Let $\bm{x}$ and $\mb{\Gamma}_x$ be a stochastic signal on $\mathcal{G}$ and its covariance matrix, respectively, by Definition \ref{definition:gwss2}. Then, $\mb{\Gamma}_x$ is diagonalizable by $\mb{U}$.
\end{lemma}

\begin{proof}\label{proof:gwssm}
The LHS of \eqref{cond2:gwss2} is expressed as
\begin{align}
&\mathbb{E}[(\bm{x}\star \bm{\delta}_n-\mu_x\bm{1})_m(\bm{x}\star \bm{\delta}_k-\mu_x\bm{1})_m^*]\nonumber\\
&=[\mb{M}\mathrm{diag}(u_{0}[n],u_{1}[n],\ldots)\widehat{\mb{\Gamma}}_x\mathrm{diag}(u^*_{0}[k],u^*_{1}[k],\ldots)\mb{M}^*]_{m,m}\nonumber\\
&=\bm{\delta}_m^*\mb{M}\mathrm{diag}(\mb{U}^*\bm{\delta}_n)\widehat{\mb{\Gamma}}_x\mathrm{diag}(\mb{U}^*\bm{\delta}_k)\mb{M}^*\bm{\delta}_m\nonumber\\
&=\bm{\delta}_n^*\mb{U}\mathrm{diag}(\mb{M}^*\bm{\delta}_m)\widehat{\mb{\Gamma}}_x\mathrm{diag}(\mb{M}^*\bm{\delta}_m)\mb{U}^*\bm{\delta}_k\nonumber\\
&=[\mb{U}(\widehat{\mb{\Gamma}}_x\circ \mb{M}^*\bm{\delta}_m\bm{\delta}_m^*\mb{M})\mb{U}^*]_{n,k}\nonumber\\
&=[\mb{U}(\widehat{\mb{\Gamma}}_x\circ \mb{\Xi})\mb{U}^*]_{n,k},\label{proof:rescov1}
\end{align}
where $\widehat{\mb{\Gamma}}_x = \mathbb{E}[\mb{U}^*(\bm{x}-\mu_x\bm{1})(\bm{x}-\mu_x\bm{1})^*\mb{U}]$ and $[\mb{\Xi}]_{i,l}=[\mb{M}^*\bm{\delta}_m\bm{\delta}_m^*\mb{M}]_{i,l}=\exp(-j2\pi (i-l)/N)$. 
Since $[\mb{\Xi}]_{i,l}=1$ for $i=l$ and $[\mb{\Xi}]_{i,l}\neq 1$ otherwise, the condition satisfying \eqref{proof:rescov1}, i.e., $\widehat{\mb{\Gamma}}_x=\widehat{\mb{\Gamma}}_x\circ\mb{\Xi}$, holds if and only if $\widehat{\mb{\Gamma}}_x$ is diagonal. Therefore, $\mb{\Gamma}_x$ is diagonalizable by $\mb{U}$, which completes the proof.
\end{proof}

\subsection{Relationship among GWSS definitions}\label{app:relation_gwss}

We now discuss the relationship among some representative definitions of GWSS from the viewpoint of the PSD. These results are summarized in Table \ref{app:char_psd}.
Existing definitions of GWSS are defined as counterparts of the classical WSS definitions in Corollaries \ref{cor:wss_shift} and \ref{cor:wss_localize}.
In terms of diagonalizability, GWSS$_{\text{M}}$ (the counterpart of Corollary \ref{cor:wss_shift}) is stricter than GWSS$_{\text{T}}$ (the counterpart of Corollary \ref{cor:wss_localize}).
This is because the covariance is always diagonalizable by $\mb{U}$ in GWSS$_{\text{M}}$, while that is not the case with GWSS$_{\text{T}}$ in general.

Next, we compare $\mb{\Theta}$ in \eqref{proof:rescov2} with $\mb{\Xi}$ in \eqref{proof:rescov1}. In \eqref{proof:rescov1}, off-diagonal entries in $\mb{\Xi}$ are not equal to 1. In contrast, 
those in $\mb{\Theta}$ can take the value 1 in the case $\lambda_i=\lambda_l$ for $i\neq l$, i.e., eigenvalues with multiplicity greater than 1. Therefore, GWSS$_\text{T}$ allows the existence of non-zero off-diagonal elements of $\widehat{\mb{\Gamma}}_x$ for some graphs having repeated eigenvalues, while GWSS$_\text{M}$ always yields the diagonal $\widehat{\mb{\Gamma}}_x$ regardless of graphs.
In fact, GWSS$_{\text{M}}$ and GWSS$_{\text{T}}$ coincide with each other if all eigenvalues of $\mb{L}$ are distinct (cf. Lemma \ref{lemma:gwssm}).
Therefore, GWSS$_{\text{T}}$ effectively assumes distinct eigenvalues of $\mb{L}$ in \cite{girault_stationary_2015}.

It is often assumed that $\mb{\Gamma}_x$ is a polynomial in $\mb{L}$ \cite{perraudin_stationary_2017,segarra_stationary_2017}\footnote{In this paper, we simplify a counterpart of time-shift by means of $\mb{L}$ for simplicity. Nevertheless, we can easily extend the proposed GWSS for other graph variation operators, including $\mb{A}$ and $\bm{\mathcal{L}}\coloneqq\mb{I}-\mb{D}^{-1/2}\mb{A}\mb{D}^{-1/2}$.}. This is  sufficient for the diagonalizability of $\mb{\Gamma}_x$.
It is noteworthy that the polynomial assumption is equivalent to $\widehat{\Gamma}_x(\lambda_i)=\widehat{\Gamma}_x(\lambda_l)$ for all $\lambda_i=\lambda_l$\cite{perraudin_stationary_2017}.
This is a special case of GWSS$_\text{M}$, since GWSS$_\text{M}$ allows for  $\widehat{\Gamma}_x(\lambda_i)\neq\widehat{\Gamma}_x(\lambda_l)$ for any $\lambda_i=\lambda_l$.
As a result, GWSS$_\text{M}$ is a good compromise between applicability and a rigorous relationship to the WSS definition.
Therefore, we use GWSS$_\text{M}$ as our GWSS definition.

\begin{table}[!t]
\setlength{\tabcolsep}{0.33em}
\centering
\caption{Diagonalizability of the covariance for GWSS definitions.
If all eigenvalues in $\mb{\Lambda}$ are distinct, they coincide with each other.}\label{app:char_psd}
\begin{tabular}{c|c}
\hline
\begin{tabular}[c]{@{}c@{}}$\bm{\Gamma}_x$ is not necessarily\\  diagonalizable by $\mathbf{U}$\end{tabular} &GWSS$_\text{T}$\cite{girault_stationary_2015}  \\ \hline
  $\bm{\Gamma}_x$ is diagonalizable by $\mathbf{U}$  & GWSS$_\text{M}$ \\\hline
  $\bm{\Gamma}_x$ is  a polynomial in $\mathbf{L}$ & \cite{perraudin_stationary_2017,segarra_stationary_2017} \\ \hline
\end{tabular}
\end{table}

\end{appendices}

\bibliographystyle{ieeetr}
\label{sec:refs}
\bibliography{gen_graph_samp_jurnal_cleaned} 

\end{document}